\documentclass[12pt, draftclsnofoot, onecolumn]{IEEEtran}
\usepackage[super]{nth}
\usepackage{amsthm}
\usepackage{algorithm}
\usepackage{amssymb}
\usepackage{multirow}

\usepackage{tabularx}  
\usepackage{arydshln,leftidx,mathtools} 

\theoremstyle{definition}
\newtheorem{theorem}{Theorem}

\newtheorem{lemma}{Lemma}

\usepackage{cite}
\ifCLASSINFOpdf
\usepackage[pdftex]{graphicx}
\else
\fi

\usepackage{amsmath}
\usepackage{array}
\usepackage{courier}
\usepackage[usenames,dvipsnames]{color} 
\usepackage{color,url,xspace}
\usepackage{amssymb,amsmath}
\usepackage{graphicx}
\usepackage{colordvi}
\usepackage{epsfig}
\usepackage{endnotes}
\usepackage{verbatim}
\usepackage{textcomp}
\usepackage{setspace}
\usepackage{xcolor}
\usepackage{cite}
\usepackage{times}
\usepackage{booktabs, multirow}
\usepackage{amsthm}
\usepackage{nccmath}
\usepackage{algorithmicx,algpseudocode}
\usepackage{algpseudocode,algorithm,algorithmicx}
\usepackage{caption}
\usepackage{mathtools}
\usepackage{subcaption}

\algrenewcommand\algorithmicrequire{\textbf{Precondition:}}
\algrenewcommand\algorithmicensure{\textbf{Postcondition:}}

\begin{document}
	
	\title{Blind Signal Classification Analysis and Impact on User Scheduling and Power Allocation in Nonorthogonal Multiple Access}
	
	\author{
		Minseok Choi, \IEEEmembership{Member,~IEEE,}
		and Joongheon Kim, \IEEEmembership{Senior Member,~IEEE}
		\thanks{M. Choi is with the School of Electrical Engineering, Korea Advanced Institute of Science and Technology, Daejeon, Korea (e-mail: ejaqmf@kaist.ac.kr).} 
		\thanks{J. Kim is with the School of Computer Engineering, Chung-Ang University, Seoul, Korea (e-mail: joongheon@cau.ac.kr).} 
		\thanks{J. Kim is the corresponding author of this paper.} 
	}
	
	\maketitle
	
	\begin{abstract}
		
		For a massive number of devices, nonorthogonal multiple access (NOMA) has been recognized as a promising technology for improving the spectral efficiency compared to orthogonal multiple access (OMA). 
		However, it is difficult for a base station (BS) to provide all of the information about NOMA signals via a high layer owing to signaling overhead concerns. 
		This paper studies blind signal classification, which determines whether or not the received NOMA signal requires successive interference cancellation (SIC) without a priori signal information.
		In this paper, two types of blind signal classification errors are analyzed: 
		1) the signal is classified as one that does not require SIC on the SIC user side and 2) the signal is classified as one for which SIC is necessary on the non-SIC user side.
		In addition, we formulate the joint optimization problem for user scheduling and power allocation, which maximizes the sum-rate gain of NOMA over OMA with constraints on the maximum classification error probability and minimum data rate.
		The proposed algorithm iteratively finds solutions for user scheduling and power allocation.
		Simulation results show that the proposed scheme outperforms existing user scheduling methods.
		
	\end{abstract}
	
	
	\IEEEpeerreviewmaketitle
	
	\section{Introduction}
	\label{sec:Introduction}
	
	In next-generation communications, it is necessary to provide high data rates to a large number of heterogeneous wireless devices with limited resources~\cite{5G:Andrews:JSAC2014}.
	To achieve this goal, nonorthogonal multiple access (NOMA) has recently emerged as a promising technology to improve both the efficiency of resource utilization and the system performance in 5G networks~\cite{NOMA_5G:Dai:CommMag2015},\cite{NOMA-SystemPerf:Saito:PIMRC}. 
	Recently, 3GPP has also studied deployment scenarios and receiver designs for NOMA systems in Release 14 in the context of 
	a work item called ``multiple user superposition transmission" (MuST) \cite{MuST:3GPP}.
	Orthogonal multiple access (OMA), mainly employed in 3G or 4G networks, allocates orthogonal resources to user terminals, thereby eliminating inter-user interference.
	On the other hand, NOMA superposes multiuser signals with the different power weights within the same frequency, time, and spatial domains.
	Thus, NOMA receivers should handle the interference from the superpositioned signals by successive interference cancellation (SIC) or maximum-likelihood (ML) detection.
	NOMA can provide significant benefits 
	for cell throughput improvement compared to OMA with the assumption of ideal SIC~\cite{Book:Tse}.
	
	NOMA has been widely developed with other technologies in various environments.
	The application of MIMO to NOMA has also been investigated~\cite{MIMO-NOMA:Ding:TWC2016,MIMO-NOMA:Ding:TWC2016-2}, and the ergodic capacity of MIMO-NOMA has been derived~\cite{MIMO-NOMA:Sun:WCL2015}.
	In addition, the data rate of a cell-edge user can be increased by using NOMA in a cooperative system \cite{NOMA-CoMP:Choi:CL2014} or a distributed antenna system \cite{NOMA-CoMP:ICC-Han}.
	Cooperative NOMA has been proposed in an environment where cooperation among users is possible~\cite{NOMA-Coop:Ding:CL2015,NOMA-Coop:Choi:ArXiv2018}. 
	The fairness among the scheduled NOMA users is also studied in \cite{NOMA:UserFairness:SPL-Timotheou}.

	Before an attempt is made to handle interference in NOMA signals on the receiver side, the receiver has to know some signal information in advance, e.g., 1) signal identification, i.e., whether the signal is modulated by OMA or NOMA; 2) modulation classification;
	and 3) whether or not SIC is required for the received signal \cite{BlindSC:ArXiv2018Choi}.
	All of this signal information can be transmitted to the receiver via a high layer. 
	However, the required signaling overhead can be a concern, especially in cellular environments with a massive number of heterogeneous devices, e.g., Internet of Things (IoT) networks.
	Further, some high-layer signaling may or may not be used on the receiver side depending on the multiuser channel conditions. 
	In this case, high-layer signaling 
	for NOMA signal information can be an extremely large waste of valuable resources.
	This motivates blind signal classification on the receiver side.
	3GPP has also discussed blind signal classification in NOMA systems 
	\cite{3GPP-TSG-RAN:Intel,3GPP-TSG-RAN:MediaTek}.
	However, to the best of the authors' knowledge, blind signal classification in NOMA systems has not been thoroughly investigated in a theoretical sense.
	
	On the basis of an ML-based classifier, which has been researched for OMA for a long time \cite{MC:TC2000Wei,MC:TWC2009Hameed,MC:CL2013Soltanmohammadi}, the performance of signal identification and modulation classification improves as the SNR increases.
	Therefore, the system is sufficient for scheduling users whose SNRs are larger than a certain threshold for guaranteeing reliable signal identification and blind modulation classification.
	As the reader will see in detail later, however, the performance of blind signal classification for the presence of interference even decreases with the SNR for the user who does not perform SIC.
	That is because the user who does not perform SIC has to determine that there is no interference in the received signal, even though the signals of multiple users are superposed at the transmitter.
	In this regard, more elaborate user scheduling and power allocation are necessary in a NOMA system where users perform blind signal classification for the presence of interference.
	Thus, this paper focuses on determining whether or not SIC is required on the user side, which can be regarded as blind signal classification for the presence of interference.
	
	User scheduling and resource allocation for NOMA have been extensively researched.
	The impact of user pairing on NOMA transmissions in a hybrid multiple access system, which allows both OMA and NOMA users, has been researched in \cite{NOMA-UserSchedule:Ding:TVT2016}, and the optimal user pairing for downlink NOMA was proposed in \cite{NOMA-US:WCL2018Kang}.
	The power allocation problem has also been extensively researched with fixed user pairing for NOMA \cite{NOMA:UserFairness:SPL-Timotheou,NOMA-PowerAllocation:Yang:TWC2016,NOMA-PowerAllocation:Choi:TWC2016,NOMA-PowerAllocation:Choi:CL2016,NOMA-PowerAllocation:Di:TWC2016,NOMA-ResourceAllocation:Fang:TC2016,NOMA-ResourceAllocation:Zhu:JSAC2017,NOMA-PA:SPL2016Cui}, but user scheduling is not considered.
	The authors of \cite{NOMA-PowerAllocation:Di:TWC2016,NOMA-ResourceAllocation:Fang:TC2016,NOMA-ResourceAllocation:Zhu:JSAC2017} focus on joint optimization of the subchannel assignment and power allocation but consider the situation where the base station (BS) already determines which users will be served by NOMA; thus, user scheduling for NOMA transmissions in a hybrid multiple access system has not been investigated.
	The joint optimization of power allocation and user scheduling for a NOMA system has been researched in \cite{NOMA-US:TC2017Liang, NOMA-US:TWC2018Cui, NOMA-US:WCL2018Zhu, ArXiv2018Choi}.
	The distributed matching algorithm was used for the optimal two-user schedulings and power allocation in \cite{NOMA-US:TC2017Liang}.
	The authors of \cite{NOMA-US:TWC2018Cui} studied the application of NOMA in millimeter wave communications, and the globally optimal two-user pairings and power allocation were studied in \cite{NOMA-US:WCL2018Zhu}.
	In \cite{ArXiv2018Choi}, a dynamic algorithm for user scheduling and power allocation is presented to maximize the sum rate while reducing the queuing delay.
	However, all of the above studies do not consider blind signal classification in an NOMA system.
	
	The main contributions of this paper are as follows:
	\begin{itemize}
		\item Two types of incorrect blind signal classification for the presence of interference in two-user NOMA are analyzed. 
		In the first type, the SIC user 
		determines that it should not perform SIC, and in the second type, 
		the non-SIC user classifies itself as the SIC user. 
		In addition, mathematical forms of the probabilities for these two errors are derived.
		
		\item The effects of blind signal classification for the presence of interference on user scheduling and power allocation in a hybrid multiple access system are investigated.
		The trade-off between the data rate and the classification error probability is studied, and we show that appropriate user scheduling and power allocation can control the trade-off.
		
		\item The joint optimization problem of user scheduling and power allocation in a hybrid multiple access system is formulated.
		We then solve the optimal power allocations given user scheduling, and iteratively find the appropriate user scheduling.
		
		\item We numerically investigate how many data samples should be used to guarantee reliable blind signal classification for the presence of interference.
		
		\item Our extensive simulation results show that the proposed scheme outperforms other existing user scheduling methods.
	\end{itemize}
	
	The rest of this paper is organized as follows. 
	In Section \ref{sec:system_model}, the system model is described.
	Section \ref{sec:ergodic_capacity_analysis} derives the mathematical forms of the classification error probabilities for the presence of interference.
	The joint optimization problem of user scheduling and power allocation is formulated, and an algorithm is proposed to solve this problem in Section \ref{sec:power_allocation}.
	Simulation results are presented in Section \ref{sec:simulation}.
	Finally, the conclusion follows in Section \ref{sec:conclusion}.
	
	$P\{.\}$ and $p(.)$ represent the probability of event occurrence and the probability density function of a random variable, respectively.
	
	\section{System Model}
	\label{sec:system_model}
	
	\subsection{Cellular Model and Downlink NOMA Transmission}
	Consider downlink communications in a cellular model where a BS transmits signals to $K$ users simultaneously on the basis of hybrid multiple access. 
	We focus on the downlink scenario without any high-layer signaling that communicates the necessity of SIC use on the NOMA user side.
	Therefore, NOMA users should perform blind signal classification for the presence of interference.
	Suppose that users $k$ and $n$ are scheduled for NOMA at the BS, $k,n\in\{1,\cdots,K\}$.
	The BS intentionally superposes the signals for both users with different power weights; thus, the signal received by user $i$ is given by 
	\begin{equation}
	y_i = h_i(s_k + s_n) + w_i,
	\end{equation}
	where $y,s,h,$ and $w$ correspond to the received signal, transmitted symbol, channel, and noise, respectively, and the subscripts $k,n,$ and $i\in \{k,n\}$ indicate the users' indices. 
	Denote $\gamma_k$ and $\gamma_n$ as the power allocation ratios for users $k$ and $n$, respectively; therefore,  $\mathrm{E}[|s_k|^2] = \gamma_k$. 
	In addition, a normalized power is assumed; i.e., $\gamma_k + \gamma_n = 1$. 
	Let $\chi_k$ and $\chi_n$ be the power weighted constellation sets of users $k$ and $n$, respectively; then, $s_k \in \chi_k$ and $s_n \in \chi_n$.
	Moreover, $s=s_k+s_n \in \chi = \chi_k \oplus \chi_n$, where $\chi$ is the composite constellation set of the superpositioned NOMA signal.
	$\chi = \chi_k \oplus \chi_n$ means that the set $\chi$ consists of the sums of all combinations of elements in $\chi_k$ and $\chi_n$.
	
	The Rayleigh fading channel from the BS to user $i$ is defined as $h_i = \sqrt{L_i}g_i$ for $i=1,\cdots,K$, where $L_{i}=1/d_i^2$ controls the path loss; $d_i$ is the BS-user-$i$ distance; and $g_i$ represents the fast fading component having a complex Gaussian distribution, $g_i \sim CN(0,1)$.
	Without loss of generality, assume that $|h_k|^2 \leq |h_n|^2$ and $|h_1|^2 \geq |h_2|^2 \geq \cdots \geq |h_K|^2$.
	In addition, suppose that the BS knows the instantaneous channel gains, and $w_i \sim CN(0,\sigma^2)$, where $\sigma^2$ is the normalized noise variance.
	
	\begin{figure} [t!]
		\centering
		\includegraphics[width=0.5\textwidth]{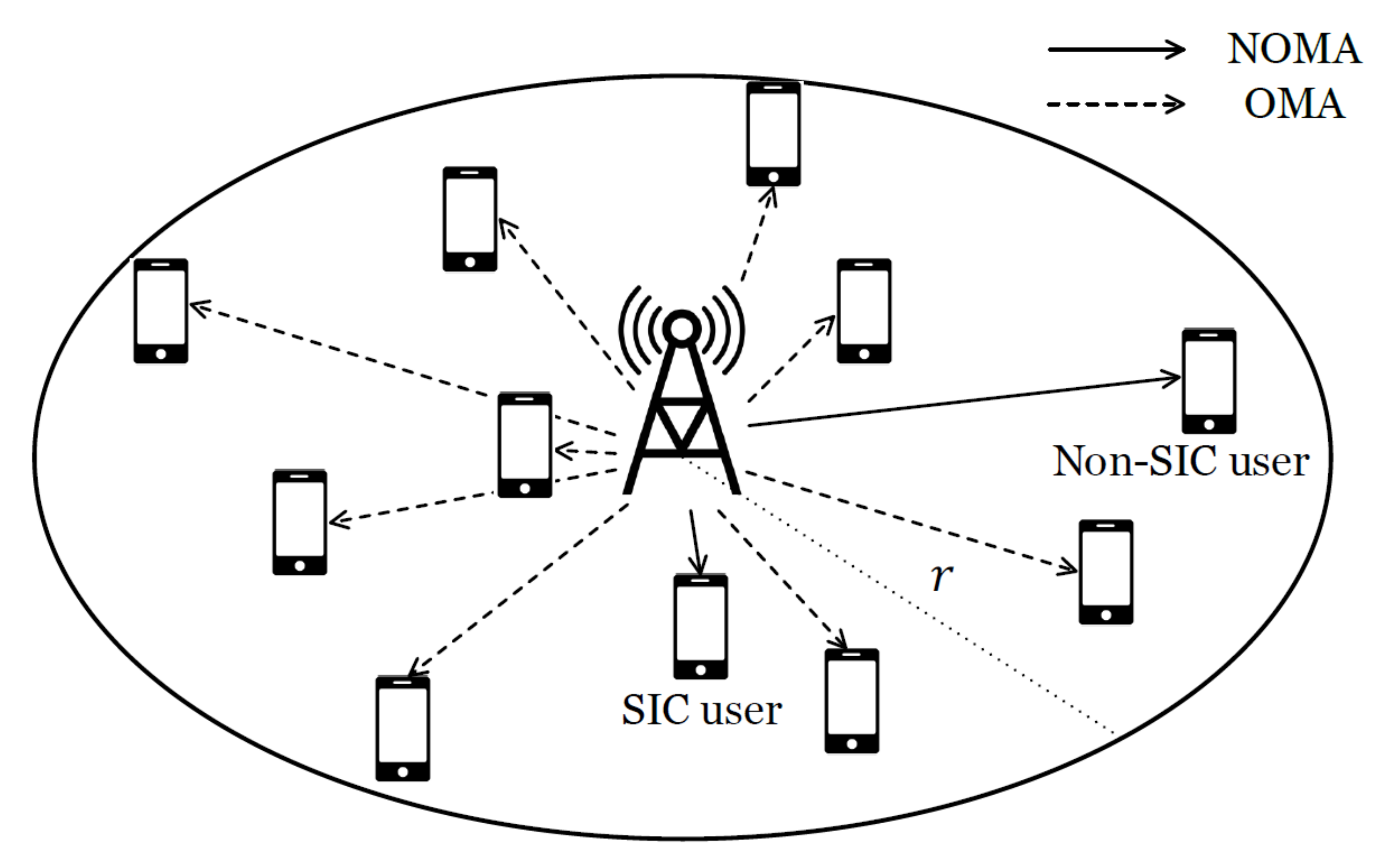}
		\caption{Cellular network model of two-user NOMA and multiuser OMA}
		\label{Fig:SystemModel}
	\end{figure}
	
	A larger power is usually allocated to the user with weak channel conditions in the NOMA system, $\gamma_k>\gamma_n$.
	With a large power allocation, user $k$ does not perform SIC and just decodes $s_k$ without cancellation of $s_n$. 
	Therefore, the data rate of user $k$ is given by
	\begin{equation}
	R_k = \log_2 \Big(1 + \frac{|h_k|^2\gamma_k}{|h_k|^2\gamma_n + \sigma^2}\Big).
	\label{eq:k-user data rate}
	\end{equation}
	Meanwhile, SIC is necessary for user $n$ to cancel user $k$'s signal component, $s_k$, and its SINR is $\frac{|h_n|^2\gamma_k}{|h_n|^2\gamma_n + \sigma^2}$. 
	However, $R_k$ remains the same as (\ref{eq:k-user data rate}) because $|h_k|^2 < |h_n|^2$.
	After performing SIC, the data rate of user $n$ becomes
	\begin{equation}
	R_n = \log_2 \Big(1 + \frac{|h_n|^2 \gamma_n}{\sigma^2}\Big).
	\label{eq:n-user data rate}
	\end{equation}
	
	OMA is the general baseline multiple access scheme; therefore, the fair allocation of frequency resources is reasonable in a hybrid multiple access system.
	Therefore, since a normalized frequency resource is assumed for two NOMA users, the data rate of user $i$ served by OMA is given by 
	\begin{equation}
	\tilde{R}_i = \frac{1}{2}\log_2 \Big( 1 + \frac{|h_i|^2}{\sigma^2} \Big).
	\end{equation}
	
	Since (\ref{eq:k-user data rate}) and (\ref{eq:n-user data rate}) depend on the power coefficients, finding the optimal power allocation ratios is as important as appropriate user scheduling.
	Thus, this paper considers the joint optimization of user scheduling and power allocation for an NOMA system to maximize the sum rate when all users perform blind signal classification.
	In a hybrid multiple access system, several user groups consisting of more than two users can be created for power-multiplexing NOMA.
	However, this paper considers only two-user groupings for NOMA transmissions and assumes that other users are served by OMA, owing to the excessive complexity of blind signal classification, as we will see later.
	
	Here, the problem is how to schedule two users for NOMA transmissions among $K$ users while guaranteeing reliable blind signal classification for the presence of interference without any high-layer signaling.
	The multiple groups for NOMA can be generated in a hybrid multiple access system, but this paper focuses on scheduling only two users for NOMA transmissions among $K$ users, for simplicity. 
	The joint optimization of multiple user groupings and power allocations under the constraint of reliable blind signal classification is exceedingly complicated.
	In addition, even though it is not optimal, a suboptimal user scheduling policy can be found by iterative methods, e.g., the deferred acceptance procedure for stable matching \cite{deferred_acceptance}.
	In this paper, however, we leave the optimal method of making multiple user groups for NOMA as a future work.
	
	\subsection{ML-Based Blind Signal Classification}
	
	The concept of ML-based modulation classification \cite{MC:TC2000Wei} is applied to blind signal classification for the existence of interference in this paper.
	For simplicity, we call the user for which SIC is necessary and the user who does not perform SIC as the SIC user and non-SIC user, respectively.
	Let two hypotheses, $\mathcal{H}_S$ and $\mathcal{H}_N$, represent the cases where the user who receives the signal is the SIC user and the non-SIC user, respectively.
	The likelihood probabilities of the SIC user and non-SIC user are then given by 
	\begin{align}
	p(y|\mathcal{H}_S) &= \sum_{s' \in \chi} p(y|s') p(s') = \frac{1}{|\chi|} \sum_{s' \in \chi}{\frac{1}{\pi \sigma^2}e^{-\frac{|y - hs'|^2}{\sigma^2}}}
	\label{eq:ML_SICuser}\\
	p(y|\mathcal{H}_N) &= \sum_{s' \in \chi_k} p(y|s') p(s') =  \frac{1}{|\chi_{k}|} \sum_{s' \in \chi_{k}}{ \frac{1}{\pi \sigma^2}e^{-\frac{|y - hs'|^2}{\sigma^2}} },
	\label{eq:ML_nonSICuser}
	\end{align}
	respectively, where the symbol $s'$ is assumed to be equally probable, and the cardinalities of the constellation sets $\chi$ and $\chi_k$ are denoted by $|\chi|$ and $|\chi_k|$, respectively.
	Since the SIC user should detect both superposed signals, $s_k$ and $s_n$, $p(y|\mathcal{H}_S)$ in \eqref{eq:ML_SICuser} is calculated throughout the composite constellation $\chi$.
	On the other hand, the non-SIC user detects only $s_k$; therefore, it is enough to scan the constellation set of user $k$, i.e., $\chi_k$, for obtaining \eqref{eq:ML_nonSICuser}.
	If one sample of received signals, i.e., $y$, is used for blind signal classification and $p(y|\mathcal{H}_S) > p(y|\mathcal{H}_N)$, the receiver classifies itself as the SIC user; otherwise, it classifies itself as the non-SIC user.
	
	Note that since the received signal $y$ is originally generated from $\chi$, $p(y|\mathcal{H}_S)$ increases, but $p(y|\mathcal{H}_N)$ decreases as the SNR grows. 
	Therefore, it is beneficial to schedule the user with a high SNR and the user with a low SNR as the SIC user and the non-SIC user, respectively, for reliable blind signal classification for the presence of interference.
	In this case, however, there is a risk that the data rate of the non-SIC user would not be sufficiently large to satisfy the minimum data rate constraint.
	Thus, this paper considers two conflicting constraints for the non-SIC user, the minimum data rate and the maximum error probability of blind signal classification for the presence of interference.
	
	\textit{Remark}: In practice, the classification steps for whether the received signal is based on OMA or NOMA, i.e., signal identification, and which modulation scheme is employed, i.e., blind modulation classification, should also be considered \cite{BlindSC:ArXiv2018Choi}. 
	The performance of the ML-based classifiers for those steps increases with the SNR on every NOMA user side \cite{MC:TC2000Wei}; therefore, the BS needs to schedule just the users with sufficiently large SNRs to guarantee the reliability of signal identification and blind modulation classification. 
	Meanwhile, the performance of blind classification for the presence of interference at the non-SIC user decreases with the SNR.
	As in \eqref{eq:ML_nonSICuser}, the received symbols should be closer to $\chi_k$ than to $\chi$, in order for the non-SIC user to determine that it should not perform SIC. 
	However, since the received signal is superposed and generated on the basis of $\chi$ at the transmitter side, the received symbols become close to $\chi$ as the SNR increases, as shown in Fig. \ref{subfig:highSNR}. 
	At rather low SNRs, it is highly probable that the non-SIC user classifies itself correctly, because more received symbols are closer to $\chi_k$ than to $\chi$ compared to the high-SNR scenario.
	In this regard, more elaborate user scheduling and power allocation are necessary in an NOMA system where users perform blind signal classification for the presence of interference.
	Accordingly, we only consider blind signal classification for the presence of interference in this paper.
	From now on, the term ``blind signal classification" means the determination of whether or not SIC is required for decoding the received signal on the user side.
	
	\begin{figure}[h]
		\centering
		\begin{subfigure}[b]{0.45\textwidth}
			\includegraphics[width=\textwidth]{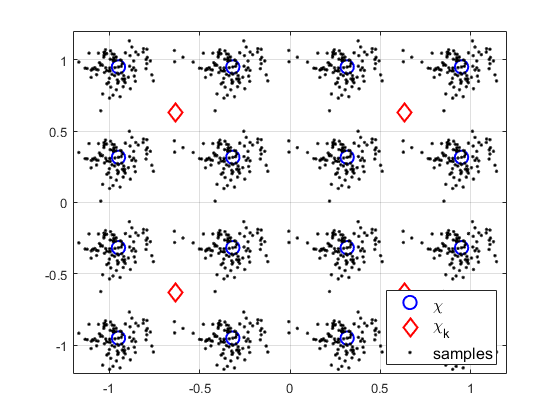}
			\caption{High SNR}
			\label{subfig:highSNR}
		\end{subfigure}
		\begin{subfigure}[b]{0.45\textwidth}
			\includegraphics[width=\textwidth]{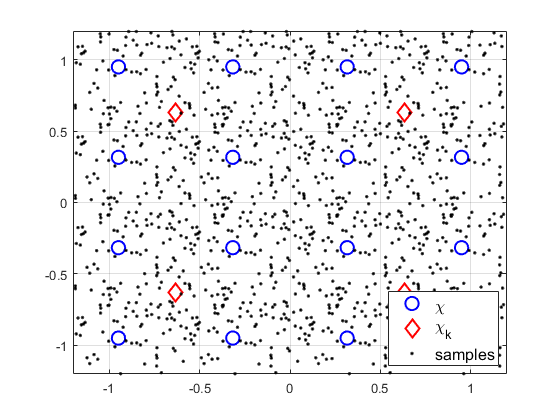}
			\caption{Low SNR}
			\label{subfig:lowSNR}
		\end{subfigure}
		\caption{Received NOMA symbols and constellation points}\label{fig:const}
	\end{figure}
	
	\subsection{Extension to $M$-User Grouping for NOMA}
	
	Although this paper mainly focuses on a two-user NOMA system, blind signal classification can be performed in the general $M$-user NOMA model.
	Again, $|h_1|^2 \geq |h_2|^2 \geq \cdots \geq |h_M|^2$; then, the largest power is allocated to user $M$, and user $M$ 
	decodes $s_M$ without cancellation of any interfering signal.
	By contrast, user 1 should cancel all other users' signals.
	In general, user $m$ decodes its data after canceling the interference components of user $M$, user $M-1$, $\cdots$, user $m+1$ in order.
	Denote $\chi_{n-m}$ as the composite constellation set of $\chi_n$, $\chi_{n+1}$, $\cdots$, and $\chi_m$, i.e., $\chi_{n-m} = \chi_n \oplus \chi_{n+1} \oplus \cdots \oplus \chi_m$. 
	Let $\mathcal{H}_m$ be the hypothesis that indicates that the target user is user $m$; i.e., $M-m-1$ SIC steps are required.
	The likelihood probability of user $m$ is given by 
	\begin{equation}
	p(y|\mathcal{H}_m) = \frac{1}{|\chi_{(m+1)-M}|} \sum_{s \in \chi_{(m+1)-M}} \frac{1}{\pi \sigma^2} e^{-\frac{|y-hs|^2}{\sigma^2}}.
	\end{equation}
	The receiver then determines itself to be user $m_0$, where
	\begin{equation}
	m_0 = \underset{m\in \{1,\cdots,M \} }{\arg \max}~p(y|\mathcal{H}_m).
	\end{equation}
	However, as the number of hypotheses increases, i.e., the number of NOMA users increases, the classification performance is expected to be significantly degraded.
	In addition, a massive number of computations is required to obtain the likelihood probabilities of all users; therefore, blind signal classification is not preferred when a large number of users is served by NOMA. 
	Thus, it is reasonable to focus on a two-user NOMA system where blind signal classification for the presence of interference is performed on the user side.
	
	\begin{table}[t!]
		\caption{Example of MCS table for NOMA}
		\label{table:MCS}
		\begin{center}
			\begin{tabular}{|c|c|c|c|c|c|c|}
				\hline
				{} & \multicolumn{3}{|c|}{SIC user} & \multicolumn{3}{|c|}{non-SIC user} \\
				\cline{2-7}
				{Index} & order & power & rate & order & power & rate \\
				\hline
				1 & $m_{1}^s$ & $p_{1}^s$ & $r_{1}^s$ & $m_{1}^n$ & $p_{1}^n$  & $r_{1}^n$   \\
				\vdots & \vdots & \vdots & \vdots & \vdots & \vdots & \vdots      \\
				$M$ & $m_{M}^s$ & $p_{M}^s$ & $r_{M}^s$ & $m_{M}^n$ & $p_{M}^n$ & $r_{M}^n$     \\
				\hline
			\end{tabular}
		\end{center}
	\end{table}
	
	\section{Error Types of Blind Signal Classification for the Presence of Interference}
	\label{sec:ergodic_capacity_analysis}
	
	\subsection{Classification Error at the SIC User}
	\label{subsec:SIC->non-SIC}
	
	In this section, we focus on the probability that user $n$ (i.e., the SIC user) incorrectly determines that it should not perform SIC, denoted by $P_n\{ \hat{\mathcal{H}}_N | \mathcal{H}_S \}$, where $\hat{\mathcal{H}}_N$ indicates that the receiver classifies itself as the non-SIC user. 
	In practice, the transmitter and the receiver share the table of usable modulation and coding schemes (MCS). 
	As shown in Table \ref{table:MCS}, the MCS table of a NOMA system should include combinations of the MCSs of the SIC and non-SIC users and their power weightings. 
	Therefore, when the SIC user determines that it should not perform SIC, the signal detection would be totally incorrect even if the correct channel quality indicator (CQI) index is given.
	
	Suppose that $s_k(i_0) \in \chi_k$ and $s_n(l_0) \in \chi_n$ are the transmitted signal components for users $k$ and $n$, respectively, for $i_0 \in \mathcal{N}_k=\{1,\cdots,|\chi_k|\}$ and $k_0 \in \mathcal{N}_n = \{1,\cdots,|\chi_n|\}$. 
	The superpositioned signal becomes $s(m_0) = s_k(i_0) + s_n(k_0) \in \chi$ for $m_0 \in \mathcal{N}=\{1,\cdots,|\chi|\}$, and the received signal is $y = h (s_k(i_0)+s_n(l_0)) + w$.
	For simplicity, we define two hypotheses as follows:
	\begin{align}
	\mathcal{G}_k(i) &= \underset{s_k \in \chi_k}{\arg \max}~p(y|s_k) = s_k(i) \nonumber \\
	\mathcal{G}(i,l) &= \underset{s \in \chi}{\arg \max}~p(y|s) = s_k(i)+s_n(l) \nonumber
	\end{align}
	The classification error probability at user $n$ is then given by
	\begin{align}
	P_n\{ \hat{\mathcal{H}}_N | \mathcal{H}_S \} &= \sum_{i_0 \in \mathcal{N}_k} \sum_{l_0 \in \mathcal{N}_n} \mathrm{P} \{ p(y|\mathcal{H}_N) > p(y|\mathcal{H}_S) \} \cdot \mathrm{P} \{ s_k(i_0)+ s_n(l_0) \} 
	\label{eq:P_1-temp1} \\
	&\simeq \sum_{i_0 \in \mathcal{N}_k} \sum_{l_0 \in \mathcal{N}_n} \sum_{i_1 \in \mathcal{N}_k} \sum_{i_2 \in \mathcal{N}_k} \sum_{l \in \mathcal{N}_n} \mathrm{P} \big\{ p(y|s_k(i_1)) > p(y|s_k(i_2), s_n(l)) ~| \mathcal{G}_{k}(i_1), \mathcal{G}(i_2,l) \big\} \cdot \nonumber \\
	&~~~~~~~\frac{1}{|\chi|} \cdot \mathrm{P} \{ \mathcal{G}_{k}(i_1) \} \cdot \mathrm{P}\{ \mathcal{G}(i_2,l) \},
	\label{eq:P_1-temp2}
	\end{align}
	where all constellation points are equally probable; therefore, $\mathrm{P} \{ s_k(i_0)+ s_n(l_0) \} = 1/|\chi|$ in (\ref{eq:P_1-temp2}).
	If $s'_0=s_k(i_2)+s_n(l) \in \chi$ dominates the summation over $s' \in \chi$ in (\ref{eq:ML_SICuser}), i.e., $\underset{s \in \chi}{\arg \max}~p(y|s) = s_k(i_2)+s_n(l)$, then $p(y|\mathcal{H}_0) \approx p(y|s_k(i_2),s_n(l))$.
	In a similar way, $p(y|\mathcal{H}_1) \approx p(y|s_k(i_1))$ when $\underset{s_k}{\arg \max}~p(y|s_k) = s_k(i_0)$; then, \eqref{eq:P_1-temp1} can finally be approximated by \eqref{eq:P_1-temp2}.
	Here, $p(y|s)$ represents the probability density function of the received signal $y$ when the transmitted symbol is $s$.
	
	Since the user with a high SNR is chosen as the SIC user, we could further assume that $\underset{s_k}{\arg \max}~p(y|s_k) = s_k(i_0)$ and $\underset{s}{\arg \max}~p(y|s) = s_k(i_0)+s_n(l_0)$, which means that the transmitted symbol gives the largest likelihood probability.
	For simplicity, let the following functions represent normal distributions given only $s_k(i)$ and $s_k(i) + s_n(l)$, respectively;
	\begin{align}
	f_k(i) &= \frac{1}{\pi \sigma^2 |\chi_k|} \exp\Big( -\frac{|y-hs_k(i)|^2}{\sigma^2} \Big) \nonumber \\
	f(i,l) &= \frac{1}{\pi \sigma^2 |\chi|} \exp\Big( -\frac{|y-h(s_k(i)+s_n(l))|^2}{\sigma^2} \Big) \nonumber
	\end{align}
	Then,
	\begin{equation}
	P_n\{ \hat{\mathcal{H}}_N | \mathcal{H}_S \} \approx \sum_{i_0 \in \mathcal{N}_k} \sum_{l_0 \in \mathcal{N}_n} \mathrm{P} \{ f_k(i_0) > f(i_0,l_0) \} \cdot \frac{1}{|\chi|} 
	\end{equation}
	
	\begin{figure*}
		\begin{align}
		&P_n\{ \hat{\mathcal{H}}_N | \mathcal{H}_S \} \nonumber \\
		&~= \sum_{i_0 \in \mathcal{N}_k} \sum_{l_0 \in \mathcal{N}_n} Q \Big( \frac{\sigma^2 \ln \frac{|\chi_k|}{|\chi|} - |h|^2 (|s(m_0)|^2 - |s_k(i_0)|^2) + 2\textrm{Re}\{|h|^2 s(m_0)s_n^H(l_0)\} }{\sqrt{2\sigma^2 |hs_n(l_0)|^2}} \Big) \cdot \frac{1}{|\chi|}. \label{eq:P_1-final} \\
		&~\approx \sum_{i_0 \in \chi_{k,1}} \sum_{l_0 \in \chi_{n,1}} Q \Big( \frac{\sigma^2 \ln \frac{|\chi_k|}{|\chi|} - |h|^2 (|s(m_0)|^2 - |s_k(i_0)|^2) + 2\textrm{Re}\{|h|^2 s(m_0)s_n^H(l_0)\} }{\sqrt{2\sigma^2 |hs_n(l_0)|^2}} \Big) \cdot \frac{4}{|\chi|}. \label{eq:P_1-1stapprox}
		\end{align}
		\hrulefill
	\end{figure*}
	
	In addition, we assume that $M$-QAM is used for both users.
	Multiuser superposition transmission (MUST) adopted by 3GPP also considers only $M$-QAM for modulation candidates \cite{MuST:3GPP}. 
	Owing to the symmetric structure of $M$-QAM, we can reduce the computations required to obtain (\ref{eq:P_1-final}) by taking summations only in the first quadrature of $\chi_k$ and $\chi_n$, denoted by $\chi_{k,1}$ and $\chi_{n,1}$.
	Let the index sets of $\chi_{k,1}$ and $\chi_{n,1}$ be $\mathcal{N}_{k,1}$ and $\mathcal{N}_{n,1}$; then, $p_n(\hat{\mathcal{H}}_N | \mathcal{H}_S)$ can be further approximated in \eqref{eq:P_1-1stapprox}. 
	
	\subsection{Classification Error at the Non-SIC User}
	
	The classification error probability at user $k$ (i.e., the non-SIC user) is denoted by $P_k \{ \hat{\mathcal{H}}_S|\mathcal{H}_N \}$.
	When the non-SIC user classifies itself as the SIC user, it appears good because the interfering signal from the SIC user may be canceled, but it is not so. 
	Again, from Table \ref{table:MCS}, it is evident that the modulation order and/or power weightings of the non-SIC and SIC users could be different, and so the decision to perform SIC at the non-SIC user makes signal detection incorrect.
	
	$P_k \{ \hat{\mathcal{H}}_S|\mathcal{H}_N \}$ is obtained in a similar manner as that in Section \ref{subsec:SIC->non-SIC}:
	\begin{align}
	&P_k \{ \hat{\mathcal{H}}_S|\mathcal{H}_N \} = \sum_{i_0 \in \mathcal{N}_k} \sum_{l_0 \in \mathcal{N}_n} \mathrm{P} \{ p(y|\mathcal{H}_N) < p(y|\mathcal{H}_S) \} \cdot \mathrm{P} \{ s_k(i_0), s_n(l_0) \} 
	\label{eq:P_2-temp1} \\
	&~\simeq \sum_{i_0 \in \mathcal{N}_k} \sum_{l_0 \in \mathcal{N}_n} \sum_{i_1 \in \mathcal{N}_k} \sum_{i_2 \in \mathcal{N}_k} \sum_{l \in \mathcal{N}_n} \mathrm{P} \big\{ p(y|s_k(i_1)) < p(y|s_k(i_2), s_n(l)) ~| \mathcal{G}_k(i_1), \mathcal{G}(i_2,l) \big\} \cdot \nonumber \\
	&~~~~~~~\mathrm{P} \{ \mathcal{G}_k(i_1) \} \cdot \mathrm{P} \{ \mathcal{G}(i_2,l) \} \cdot \frac{1}{|\chi|}
	\label{eq:P_2-temp2} \\
	&\simeq \sum_{i_0 \in \mathcal{N}_k} \sum_{l_0 \in \mathcal{N}_n} \sum_{i_1 \in \mathcal{N}_k} \sum_{l\in \mathcal{N}_n} \mathrm{P} \big\{ p(y|s_k(i_1)) < p(y|s_k(i_1), s_n(l)) ~| \mathcal{G}_k(i_1), \mathcal{G}(i_1,l) \big\} \cdot \nonumber \\
	&~~~~~~~\mathrm{P}\{ \mathcal{G}_k(i_1) \} \cdot \mathrm{P} \{ \mathcal{G}(i_1,l) \} \cdot \frac{1}{|\chi|}.
	\label{eq:P_2-temp3}
	\end{align}
	
	Equation (\ref{eq:P_2-temp2}) is an approximation of \eqref{eq:P_2-temp1} under the assumptions that $s_k(i_2)+s_n(l)\in \chi$ dominates the summation over all $s'\in \chi$ in \eqref{eq:ML_SICuser}, and $s_k(i_1)$ dominates the summation over all $s' \in \chi_k$ in \eqref{eq:ML_nonSICuser}.
	Equation (\ref{eq:P_2-temp2}) is further approximated to (\ref{eq:P_2-temp3}) by $i_1=i_2$, which means that the decoding result of $s_k$ based on $\chi$ is the same as that based on $\chi_k$.
	Since the non-SIC user usually experiences weak channel conditions, the high-SNR approximation is not available.
	The first term of (\ref{eq:P_2-temp2}) can be computed as shown in \eqref{eq:P_2-1st-final}.
	
	\begin{figure*}
		\begin{align}
		&\mathrm{P} \{ f_k(i_1) < f(i_2,l) \} \nonumber \\ 
		&~= \mathrm{P} \big\{ 2\textrm{Re}\{yh^H (s_k(i_2)-s_k(i_1)+s_n(l))^H \} > \sigma^2 \ln \frac{|\chi|}{|\chi_k|} + |h|^2(|(s_k(i_2)+s_n(l))|^2 - |s_k(i_1)|^2) \big\} \nonumber \\
		&~= Q \bigg( \frac{1}{\sqrt{2\sigma^2 |h(s_k(i_2)-s_k(i_1)+s_n(l))|^2}} \bigg\{\sigma^2 \ln \frac{|\chi|}{|\chi_k|} + |h|^2(|(s_k(i_2)+s_n(l))|^2 - |s_k(i_1)|^2) \nonumber \\ 
		&~~~~~~~~~~~- 2\textrm{Re}\{|h|^2 s(m_0) (s_k(i_2)-s_k(i_1)+s_n(l))^H \bigg\} \bigg). \label{eq:P_2-1st-final}
		\end{align}
		\hrulefill
	\end{figure*}
	
	The second and third terms of (\ref{eq:P_2-temp2}) depend on the decision regions of modulation.
	In the case of $M$-QAM, all decision regions have square shapes.
	In other words, the decision region of each constellation point can be defined by two intervals of real and imaginary components.
	Denote $\mathcal{B}(s)$ as the decision region of the constellation point $s$, defined as $\mathcal{B}(s) = [b_r^-(s),b_r^+(s)] \cap [j b_i^-(s),j b_i^+(s)] \}$.
	$b_r^-(s)$ and $b_r^+(s)$ are lower and upper decision boundaries for the real component of $s$, and $b_i^-(s)$ and $b_i^+(s)$ are the lower and upper decision boundaries for the imaginary component of $s$. 
	In other words, $\mathrm{Re}\{\mathcal{B}(s)\} \in [b_r^-(s),b_r^+(s)]$ and $\mathrm{Im}\{\mathcal{B}(s)\} \in [b_i^-(s),b_i^+(s)]$.
	A matched filter is considered to remove the channel effect.
	Let $\tilde{y} = h^H y/|h|^2 = s + h^H w/|h|^2$; then the noise variance also changes to $\tilde{\sigma}^2  = \sigma^2/|h|^2$.
	Let $\tilde{y}_r$ and $\tilde{y}_i$ be the real and imaginary components of $\tilde{y}$; then, the second term of (\ref{eq:P_2-temp2}) is expressed by
	
	\begin{align}
	&\mathrm{P}\{ \mathcal{G}_k(i_1) \} = 
	\mathrm{P} \{ \tilde{y}_r \in [b_r^-(s_k(i_1)),b_r^+(s_f(i_1))] \cap \tilde{y}_i \in [b_i^-(s_k(i_1)),b_i^+(s_k(i_1))] \} \\
	&~= \int_{a_r^-(i_1)}^{a_r^+(i_1)} \frac{ e^{-{x^2}/{\tilde{\sigma}^2} }}{\sqrt{2\pi (\tilde{\sigma}^2/2)}} \mathrm{d}x \cdot \int_{a_i^-(i_1)}^{a_i^+(i_1)} \frac{ e^{-{x^2}/{\tilde{\sigma}^2}} }{\sqrt{2\pi (\tilde{\sigma}^2/2)}} \mathrm{d}x \\
	&~= \Bigg[ Q\bigg( \frac{a_r^-(i_1)}{\sqrt{\tilde{\sigma}^2/2}} \bigg) - Q\bigg( \frac{a_r^+(i_1)}{\sqrt{\tilde{\sigma}^2/2}} \bigg) \Bigg] \times \Bigg[ Q\bigg( \frac{a_i^-(i_1)}{\sqrt{\tilde{\sigma}^2/2}} \bigg) - Q\bigg( \frac{a_i^+(i_1)}{\sqrt{\tilde{\sigma}^2/2}} \bigg) \Bigg] \label{eq:P_2-2nd}
	\end{align}
	where $a_r^+(i_1) = |b_r^+(s_k(i_1)) - \textrm{Re}\{s(m_0)\}|$, $a_r^{-}(i_1) = |b_r^-(s_k(i_1)) - \textrm{Re}\{s(m_0)\}|$, $a_i^+(i_1) = |b_i^+(s_k(i_1)) - \textrm{Im}\{s(m_0)\}|$, and $a_i^-(i_1) = |b_i^-(s_k(i_1)) - \textrm{Im}\{s(m_0)\}|$.
	
	Likewise, the third term of (\ref{eq:P_2-temp2}) becomes 
	\begin{align}
	&\mathrm{P}\{ \mathcal{G}(m_2) \} = 
	\mathrm{P} \{ \tilde{y}_r \in [b_r^-(s(m_2)),b_r^+(s(m_2))] \cap \tilde{y}_i \in [b_i^-(s(m_2)),b_i^+(s(m_2))] \} \\
	&= \int_{c_r^-(m_2)}^{c_r^+(m_2)} \frac{ e^{ -{x^2}/{\tilde{\sigma}^2} }}{\sqrt{2\pi (\tilde{\sigma}^2/2)}} \mathrm{d}x \cdot \int_{c_i^-(m_2)}^{c_i^+(m_2)} \frac{ e^{ -{x^2}/{\tilde{\sigma}^2} }}{\sqrt{2\pi (\tilde{\sigma}^2/2)}} \mathrm{d}x \\
	&= \Bigg[ Q\bigg( \frac{c_r^-(m_2)}{\sqrt{\tilde{\sigma}^2/2}} \bigg) - Q\bigg( \frac{c_r^+(m_2)}{\sqrt{\tilde{\sigma}^2/2}} \bigg) \Bigg] \times \Bigg[ Q\bigg( \frac{c_i^-(m_2)}{\sqrt{\tilde{\sigma}^2/2}} \bigg) - Q\bigg( \frac{c_i^+(m_2)}{\sqrt{\tilde{\sigma}^2/2}} \bigg) \Bigg]
	\label{eq:P_2-3rd}
	\end{align}
	where $c_r^+(m_2) = |b_r^+(s(m_2)) - \textrm{Re}\{s(m_0)\}|$, $c_r^-(m_2) = |b_r^-(s(m_2)) - \textrm{Re}\{s(m_0)\}|$, $c_i^+(m_2) = |b_i^+(s(m_2)) - \textrm{Im}\{s(m_0)\}|$, and $c_i^-(m_2) = |b_i^-(s(m_2)) - \textrm{Im}\{s(m_0)\}|$.
	
	Incorporating (\ref{eq:P_2-1st-final}), (\ref{eq:P_2-2nd}), and (\ref{eq:P_2-3rd}) into (\ref{eq:P_2-temp2}), the BS could expect $P_k\{ \hat{\mathcal{H}}_S | \mathcal{H}_N \}$ on the basis of the instantaneous CSI.
	However, an extremely large computational load is required to derive (\ref{eq:P_2-temp2}).
	Since the power allocated to the non-SIC user is relatively large, if the SNR is not very low, it would be highly likely that $i_0=i_1=i_2$; i.e., symbol detection of the non-SIC user is correct, and (\ref{eq:P_2-temp2}) could be further approximated by 
	\begin{align}
	&P_k\{ \hat{\mathcal{H}}_S | \mathcal{H}_N \} \nonumber \\ 
	&\simeq \sum_{i_0\in \mathcal{N}_k} \sum_{l_0\in \mathcal{N}_n} \sum_{l\in \mathcal{N}_n} \mathrm{P} \Big\{ p(y|s_k(i_0)) < p(y|s_k(i_0), s_n(l)) |~ \mathcal{G}_n(i_0,l) \Big\} \cdot \mathrm{P}\{ \mathcal{G}_n(i_0,l) \} \cdot \frac{1}{|\chi|} 
	\label{eq:P_2-approx_bf} \\
	&\approx \sum_{i_0\in \mathcal{N}_{k,1}} \sum_{l_0\in \mathcal{N}_{n,1}} \sum_{l\in \mathcal{N}_n} \mathrm{P} \Big\{ p(y|s_k(i_0)) < p(y|s_k(i_0), s_n(l)) |~ \mathcal{G}_n(i_0,l) \Big\} \cdot \mathrm{P}\{ \mathcal{G}_n(i_0,l) \} \cdot \frac{4}{|\chi|},
	\label{eq:P_2-approx}
	\end{align}
	where $\mathcal{G}_n(i_0,l) = \underset{s_n \in \chi_n}{\arg \max}~p(y|s_k(i_0),s_n) = s_n(l)$.
	
	As in Section \ref{subsec:SIC->non-SIC}, $P_k\{ \hat{\mathcal{H}}_S | \mathcal{H}_N \}$ can be approximated by taking only the first quadratures of $\chi_f$ and $\chi_n$ owing to the symmetric structure of $M$-QAM, as given by (\ref{eq:P_2-approx}). 
	According to (\ref{eq:P_2-1st-final}) and (\ref{eq:P_2-2nd}), the first and second terms in the approximated version in (\ref{eq:P_2-approx}) can be respectively obtained by \eqref{eq:P2_1} and \eqref{eq:P2_2}, respectively, where $s(m_{i_0,l}) = s_k(i_0) + s_n(l)$.
	
	\begin{figure*}
		\begin{equation}
		Q\bigg( \frac{\sigma^2 \ln \frac{|\chi|}{|\chi_k|} + |h|^2(|s_k(i_0)+s_n(l)|^2 - |s_k(i_0)|^2) - 2\textrm{Re}\{ |h|^2 (s_k(i_0)+s_n(l_0))s_n^H(l) \} }{\sqrt{2\sigma^2 |hs_n(l)|^2}} \bigg)
		\label{eq:P2_1}
		\end{equation}
		\begin{align}
		&\Bigg[ Q\bigg( \frac{c_r^-(m_{i_0,l})}{\sqrt{\tilde{\sigma}^2/2}} \bigg) - Q\bigg( \frac{c_r^+(m_{i_0,l})}{\sqrt{\tilde{\sigma}^2/2}} \bigg) \Bigg]  \times \Bigg[ Q\bigg( \frac{c_i^-(m_{i_0,l})}{\sqrt{\tilde{\sigma}^2/2}} \bigg) - Q\bigg( \frac{c_i^-(m_{i_0,l})}{\sqrt{\tilde{\sigma}^2/2}} \bigg) \Bigg].
		\label{eq:P2_2}
		\end{align}
		\hrulefill
	\end{figure*}
	
	\subsection{Number of Data Samples Required for Blind Signal Classification}
	
	In practical environments, it is also important to determine how many data samples are required for the reliable performance of classification for the presence of interference.
	Only one data symbol is used for blind signal classification in prior sections, but more samples can provide better performance. 
	However, the use of many data samples may result in heavy signal processing tasks and a very large computational complexity.
	Assume that $L$ samples are used for blind signal classification and that all samples experience the same channel gain $h$.
	The likelihood probabilities of the SIC user and non-SIC user are then given by
	
	\begin{align}
	p(\mathbf{y}|\mathcal{H}_0) &= \frac{1}{|\chi|^L} \sum_{\mathbf{s} \in \chi^L}{\frac{1}{\pi \sigma^2}e^{-\frac{|\mathbf{y} - h\mathbf{s}|^2}{\sigma^2}}}
	\label{eq:ML_SICuser_vec}\\
	p(\mathbf{y}|\mathcal{H}_1) &= \frac{1}{|\chi_{k}|^L} \sum_{\mathbf{s} \in \chi_{k}^L}{ \frac{1}{\pi \sigma^2}e^{-\frac{|\mathbf{y} - h\mathbf{s}|^2}{\sigma^2}} },
	\label{eq:ML_nonSICuser_vec}
	\end{align}
	where $\mathbf{y}=[y_1,\cdots,y_L]$ and $\mathbf{s}=[s_1, \cdots, s_L]$ are $1 \times L$ vectors of the received signals and transmitted symbols, respectively.
	However, the numbers of summations in (\ref{eq:ML_SICuser_vec}) and (\ref{eq:ML_nonSICuser_vec}) exponentially increase with $L$; thus, a large number of computations are required to compute (\ref{eq:ML_SICuser_vec}) and (\ref{eq:ML_nonSICuser_vec}).
	
	Therefore, we independently compute the likelihood probability of each $y_l$, denoted by $p(\hat{\mathcal{H}}_i|y_l)$ for $l=1,\cdots,L$. 
	The probability that all of $L$ data samples are likely to predict $\hat{\mathcal{H}}_i$ for $i \in \{N,S\}$ becomes $\prod_{l=1}^{L} p(\hat{\mathcal{H}}_i | y_l)$.
	Since all data samples are based on the same hypothesis, it would be reasonable that $\hat{\mathcal{H}}_i$ is true if more than half of the data samples predict $\mathcal{H}_i$.
	Let $P_n\{\hat{\mathcal{H}}_S|\mathcal{H}_N,y_l\}$ and $P_k\{\hat{\mathcal{H}}_N|\mathcal{H}_S,y_l\}$ be the classification error probabilities computed from $y_l$.
	Assuming that all elements $y_l$ of $\mathbf{y}$ are independent, $P_n\{\hat{\mathcal{H}}_S|\mathcal{H}_N, y_1\} = \cdots = P_n\{\hat{\mathcal{H}}_S|\mathcal{H}_N, y_L\} = P_{n0}$ and $P_k\{\hat{\mathcal{H}}_N|\mathcal{H}_S, y_1\} = \cdots = P_k\{\hat{\mathcal{H}}_N|\mathcal{H}_S, y_L\} = P_{k0}$, because $L$ samples experience the same channel gain.
	The classification error probabilities can then be approximately computed by
	\begin{align}
	P_n\{\hat{\mathcal{H}}_S|\mathcal{H}_N\} &= \sum_{l=1}^{(L+1)/2} {L\choose l} (1-P_{n0})^l P_{n0}^{L-l},
	\label{eq:classErr_approx_near}\\
	P_k\{\hat{\mathcal{H}}_N|\mathcal{H}_S\} &= \sum_{l=1}^{(L+1)/2} {L\choose l} (1-P_{k0})^l P_{k0}^{L-l},
	\label{eq:classErr_approx_far}
	\end{align}
	with the odd number $L$.
	Since the receiver does not have to know any information about the data samples required for blind signal classification, the information bits can be used as samples for blind signal classification.
	Thus, the system does not need additional signaling overhead for blind signal classification.
	
	\section{Joint Optimization Problem of User Scheduling and Power Allocation}
	\label{sec:power_allocation}
	
	In this section, the effects of blind signal classification on user scheduling and power allocation in NOMA systems are presented. 
	For the joint optimization problem of user scheduling and power allocation, the sum-rate gain of NOMA over OMA is considered an optimization metric. 
	The gains at users $n$ and $k$ are given by
	\begin{align}
	\Delta_n &= R_n (1-P_n\{\hat{\mathcal{H}}_S|\mathcal{H}_N\}) - \tilde{R}_n, \\
	\Delta_k &= R_k (1-P_k\{\hat{\mathcal{H}}_N|\mathcal{H}_S\}) - \tilde{R}_k.
	\end{align}
	Again, users $k$ and $n$ are scheduled for NOMA as the non-SIC user and SIC user, respectively, and $|h_k|^2<|h_n|^2$. 
	The total sum-rate gain of NOMA over OMA is $\Delta_{(k,n)} = \Delta_k + \Delta_n$. 
	However, since $P_n\{\hat{\mathcal{H}}_S|\mathcal{H}_N\}$ and $P_k\{\hat{\mathcal{H}}_N|\mathcal{H}_S\}$ are too complicated to deal with, we approximate the maximization problem by maximizing the lower bound of $\Delta_{(k,n)}$.
	Let $P_t$ be the threshold of the error probability for reliable blind signal classification such that $\max \{P_{n} \{ \hat{\mathcal{H}}_N|\mathcal{H}_S \}, P_{k} \{ \hat{\mathcal{H}}_S|\mathcal{H}_N \} \} \leq P_t$.
	The lower bound of $\Delta_{(k,n)}$ is then obtained by
	\begin{equation}
	\underline{\Delta}_{(k,n)} = (R_n+R_k)(1-P_t) - \tilde{R}_n - \tilde{R}_k.
	\label{eq:Delta}
	\end{equation}
	
	The joint optimization problem of user scheduling and power allocation which maximizes the lower bound of the sum-rate gain of NOMA over OMA is formulated as follows:
	\begin{align}
	&\{n^*, k^*, \gamma^*_{n^*,(k^*,n^*)}, \gamma^*_{k^*,(k^*,n^*)}\} = \underset{n, k, \gamma_n, \gamma_k}{\arg \max}~\underline{\Delta}_{(k,n)} \label{eq:max-Delta} \\
	&~~~~\text{s.t.} \ \min \{R_{n}, R_{k}\} \cdot (1-P_t) \geq R_t, \label{eq:const-minimum-rate} \\
	&~~~~~~~~~\max \{P_{n} \{ \hat{\mathcal{H}}_N|\mathcal{H}_S \}, P_{k} \{ \hat{\mathcal{H}}_S|\mathcal{H}_N \} \} \leq P_t, \label{eq:const-maximum-errprob} \\
	&~~~~~~~~~\gamma_{n,(k,n)} + \gamma_{k,(k,n)} = 1,	\label{eq:const-power}
	\end{align}
	where $n^*$ and $k^*$ are the indices of the optimally scheduled users for NOMA. 
	$\gamma_{n^*,(k^*,n^*)}^{*}$ and $\gamma_{k^*,(k^*,n^*)}^{*}$ are the optimal power coefficients for users $n^*$ and $k^*$, respectively, when users $n^*$ and $k^*$ are scheduled for NOMA as the SIC user and non-SIC user, respectively. 
	The tightness of $\underline{\Delta}_{(k,n)}$ with respect to $\Delta_{(k,n)}$ can be controlled by $P_t$. 
	If the system takes a small $P_t$, the problem of \eqref{eq:max-Delta}--\eqref{eq:const-power} becomes almost identical to the maximization problem in $\Delta_{(k,n)}$.
	The minimum data rate of NOMA users and the reliability of blind signal classification are guaranteed by the constraints in \eqref{eq:const-minimum-rate} and \eqref{eq:const-maximum-errprob}, respectively.
	In addition, the normalized power is given by \eqref{eq:const-power}.
	
	A closed-form solution of (\ref{eq:max-Delta})--\eqref{eq:const-power} is difficult to obtain; therefore, we decouple (\ref{eq:max-Delta}) into separate problems of user scheduling and power allocation.
	First, 
	start with the arbitrary scheduling of users $k$ and $n$ as the non-SIC and SIC users, respectively, and find their optimal power allocation ratios. 
	Next, a different method of user scheduling is tested depending on the tightness of the two constraints in (\ref{eq:const-minimum-rate}) and (\ref{eq:const-maximum-errprob}).
	In this way, we will solve the problem in (\ref{eq:max-Delta})--\eqref{eq:const-power} through an iterative algorithm that obtains the optimal power allocation for fixed user scheduling and gradually replaces the scheduled users.
	
	\subsection{Power Allocation Problem}
	\label{subsec:PowerAllocationProb}
	
	Assuming that users $k$ and $n$ are already scheduled, the problem for finding the optimal power allocation is given by
	\begin{align}
	&\{\gamma_{n,(k,n)}^{*}, \gamma_{k,(k,n)}^{*}\} = \underset{\gamma_n, \gamma_k}{\arg \max}~\underline{\Delta}_{(k,n)} \label{eq:max-Delta-power} \\
	&~~~~\text{s.t.} \ \min \{R_{n}, R_{k}\} \geq \tilde{R}_t, \label{eq:const-minimum-rate-power} \\
	&~~~~~~~~~\max \{P_{n}\{ \hat{\mathcal{H}}_N|\mathcal{H}_S \}, P_{k} \{ \hat{\mathcal{H}}_S|\mathcal{H}_N \} \} \leq P_t, \label{eq:const-maximum-errprob-power} \\
	&~~~~~~~~~\gamma_{n}+\gamma_{k} = 1	\label{eq:const-power-power},
	\end{align}
	where $\tilde{R}_t = R_t/(1-P_t)$.
	
	Note that $R_k$ and $R_n$ are decreasing and increasing functions of $\gamma_n$, respectively, which is easily proved by differentiating (\ref{eq:k-user data rate}) and (\ref{eq:n-user data rate}).
	Similarly, $\underline{\Delta}_{(k,n)}$ is an increasing function of $\gamma_n$.
	Although the transmitted NOMA signal is generated from $\chi$, the classification performance of the non-SIC user depends on $\chi_k$, according to (\ref{eq:ML_nonSICuser}).
	Therefore, $P_n \{ \hat{\mathcal{H}}_N | \mathcal{H}_S \}$ can be expected to increase as the symbol points of $\chi_k$ approach those of $\chi$.
	On the other hand, $P_k \{ \hat{\mathcal{H}}_S|\mathcal{H}_N \}$ is large when $\chi$ and $\chi_k$ are clearly distinguishable from each other.
	Since a larger $\gamma_n$ results in a greater difference between $\chi$ and $\chi_k$, $P_{n} \{ \hat{\mathcal{H}}_N|\mathcal{H}_S \}$ and $P_{k} \{\hat{\mathcal{H}}_S|\mathcal{H}_N \}$ are decreasing and increasing with $\gamma_n$, respectively.
	
	Suppose that there exists a certain $\bar{\gamma}_n \in [0,1]$ that satisfies $R_k = R_n = R_0$. 
	In other words, the graphs of $R_k$ and $R_n$ for $\gamma_n \in [0,1]$ intersect each other at $\bar{\gamma}_n$ in Figs. \ref{Fig:PowerAllocation1} and \ref{Fig:PowerAllocation2}. 
	Since $R_n$ increases and $R_k$ decreases with $\gamma_n$, $R_0 \geq R_t$ is necessary for the constraint in (\ref{eq:const-minimum-rate-power}).
	Similarly, consider $\bar{\gamma}_n \in [0,1]$ such that $P_{k}\{\mathcal{H}_S|\mathcal{H}_N\} = P_{n}\{\mathcal{H}_N|\mathcal{H}_S\} = P_0$. 
	The classification error probability $P_0$ should then be smaller than $P_t$ to satisfy the constraint in (\ref{eq:const-maximum-errprob-power}).
	Assuming that $R_0 \geq R_t$ and $P_0 \leq P_t$, the following theorem gives the optimal solution of (\ref{eq:max-Delta-power})--\eqref{eq:const-power-power}.
	
	\begin{figure}[t]
		\minipage{0.45\textwidth}
		\includegraphics[width=\linewidth]{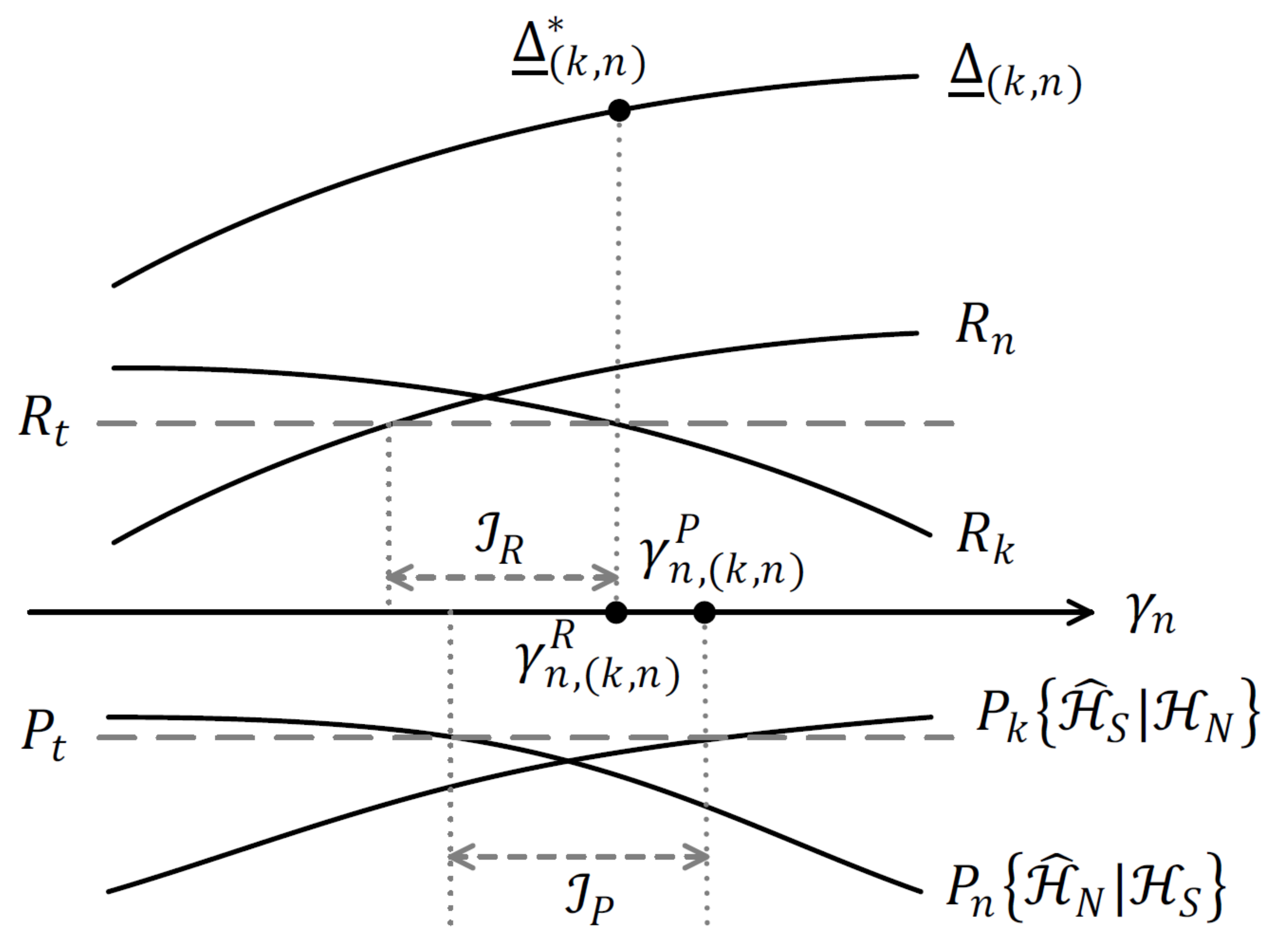}
		\caption{Power allocation example: when the minimum rate constraint is tighter than the maximum classification error constraint} 
		\label{Fig:PowerAllocation1}
		\endminipage\hfill
		\minipage{0.45\textwidth}
		\includegraphics[width=\linewidth]{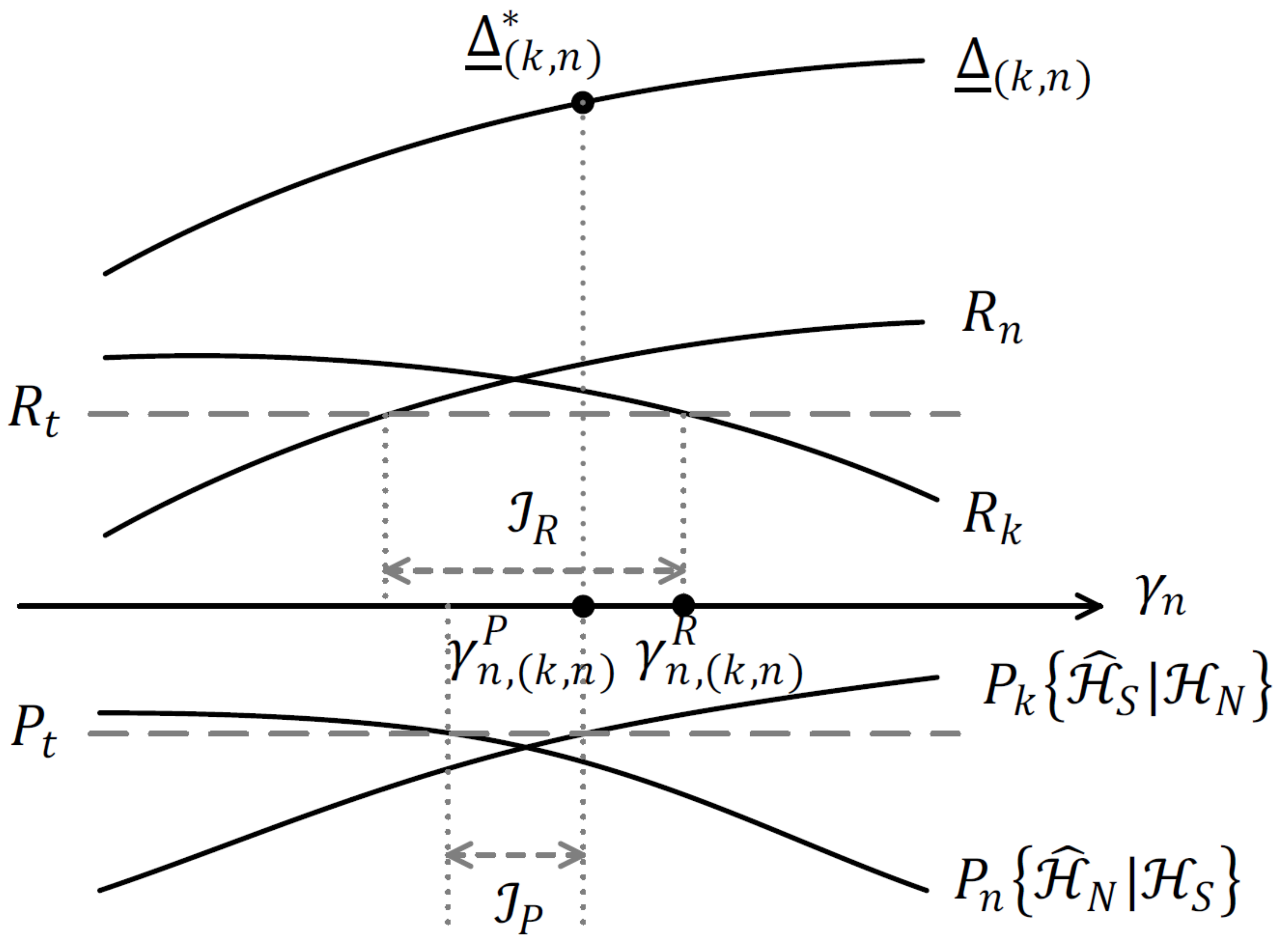}
		\caption{Power allocation example: when the maximum classification error constraint is tighter than the minimum rate constraint}
		\label{Fig:PowerAllocation2}
		\endminipage
	\end{figure}
	
	\begin{theorem}
		
		Suppose that users $k$ and $n$ are already scheduled with $|h_n|^2 > |h_k|^2$. 
		When $R_0 \geq R_t$ and $P_0 \leq P_t$ are satisfied, the optimal solution of (\ref{eq:max-Delta-power})--\eqref{eq:const-power-power} is given as
		\begin{equation}
		\gamma_{n,(k,n)}^{*} = 
		\begin{cases}
		\min\{ \gamma_{n,(k,n)}^R, \gamma_{n,(k,n)}^P \} & \text{if } \mathcal{I}_R \cap \mathcal{I}_P \neq \phi \\
		\text{No solution} & \text{otherwise}
		\end{cases},
		\label{eq:PowerAllocationSolution}
		\end{equation}
		where $\gamma_{n,(k,n)}^R$ and $\gamma_{n,(k,n)}^P$ are the power allocation ratios of user $n$ when $R_k = R_t$ and $P_k\{\mathcal{H}_0|\mathcal{H}_1\} = P_t$, respectively. 
		$\mathcal{I}_R$ and $\mathcal{I}_P$ are the intervals of $\gamma_n$ satisfying the constraints in (\ref{eq:const-minimum-rate-power}) and (\ref{eq:const-maximum-errprob-power}), respectively.
		\label{thm:optPower}
		
	\end{theorem}
	
	\begin{proof}
		
		$\underline{\Delta}_{(k,n)}$ can be shown to be an increasing function of $\gamma_n$, as given by 
		\begin{equation}
		\frac{\partial \underline{\Delta}_{(k,n)}}{\partial \gamma_n} = \bigg( \frac{|h_n|^2}{\gamma_n |h_n|^2 + \sigma^2} - \frac{|h_k|^2}{\gamma_n |h_k|^2 + \sigma^2} \bigg) \cdot \frac{1-P_t}{\ln 2}.
		\end{equation}
		$\frac{\partial \underline{\Delta}_{(k,n)}}{\partial \gamma_n} > 0$ because $|h_n|^2 > |h_k|^2$.
		First, it is obvious there is no solution when $\mathcal{I}_R \cap \mathcal{I}_P = \phi$ because
		any power allocation rule cannot satisfy both of the constraints in (\ref{eq:const-minimum-rate-power}) and (\ref{eq:const-maximum-errprob-power}).
		Consider the case $\mathcal{I}_R \cap \mathcal{I}_P \neq \phi$.
		Without the classification error constraint in \eqref{eq:const-maximum-errprob-power}, 
		because $\underline{\Delta}_{(k,n)}$ increases and $R_k$ decreases with $\gamma_n$, the upper boundary of $\mathcal{I}_R$, i.e., $\gamma^*_{n,(k,n)}$, is the optimal solution satisfying $R_k=R_t$, as shown in Figs. \ref{Fig:PowerAllocation1} and \ref{Fig:PowerAllocation2}.
		On the other hand, 
		without the minimum data rate constraint in \eqref{eq:const-minimum-rate-power}, the optimal power allocation is obtained at $P_k \{\hat{\mathcal{H}}_S|\mathcal{H}_N\} = P_t$.  
		Therefore, (\ref{eq:PowerAllocationSolution}) could be achieved to satisfy both constraints for the minimum data rate and maximum classification error probability.
		
	\end{proof}
	
	\begin{algorithm} [b!]
		\caption{Bisection method for finding $\gamma_{n,(k,n)}^P$
			\label{algo:gamma_n^P}}
		\begin{algorithmic}[1]
			\State{Initialize $\gamma_{-}=0$, $\gamma_{+}=1$.}
			\While{$\gamma_{+} - \gamma_{-} \geq \epsilon$}{
				\State{$\gamma_n^{P} = (\gamma_{+} + \gamma_{-})/2$}
				\If{$P_k \{\hat{\mathcal{H}}_S|\mathcal{H}_N\} \leq P_t$}
				$\gamma_{-} = \gamma_{n,(k,n)}^{P}$
				\Else
				~~$\gamma_{+} = \gamma_{n,(k,n)}^{P}$
				\EndIf
			}
			\EndWhile
		\end{algorithmic}
	\end{algorithm}
	
	Note that Theorem \ref{thm:optPower} can also be used to explain the fact that $\gamma_{n,(k,n)}^{*}$ is the upper boundary of $\mathcal{I}_R \cap \mathcal{I}_P$.
	For example, in Fig. \ref{Fig:PowerAllocation1}, $\gamma_{n,(k,n)}^P$ is the optimal solution for only the maximum classification error probability constraint in (\ref{eq:const-maximum-errprob}), but it does not satisfy the minimum data rate constraint in (\ref{eq:const-minimum-rate}). 
	In this case, it can be stated that the minimum data rate constraint is tighter than the maximum classification error probability constraint.
	Therefore, $\gamma^P_{n,(k,n)} > \gamma^R_{n,(k,n)}$ and $\gamma_{n,(k,n)}^{*} = \gamma_{n,(k,n)}^R$. 
	On the other hand, Fig. \ref{Fig:PowerAllocation2} shows the case in which the maximum classification error constraint is tighter than the minimum data rate constraint, where $\gamma^P_{n,(k,n)} < \gamma^R_{n,(k,n)}$.
	
	It remains to find $\gamma_{n,(k,n)}^R$ and $\gamma_{n,(k,n)}^P$.
	$\gamma_{n,(k,n)}^R$ can be directly obtained from (\ref{eq:k-user data rate}) with knowledge of the instantaneous channel gain, as given by 
	\begin{equation}
	\gamma_{n,(k,n)}^R = \frac{(|h_k|^2/\sigma^2) - (2^{R_t}-1)}{ 2^{R_t}(|h_k|^2/\sigma^2) }.
	\end{equation}
	However, a closed-form expression for $\gamma_{n,(k,n)}^P$ is difficult to derive because $P_k \{\hat{\mathcal{H}}_S|\mathcal{H}_N\}$ in (\ref{eq:P_2-approx}) is not simply expressed with $\gamma^P_{n,(k,n)}$.
	Therefore, the bisection method is presented to compute $\gamma_{n,(k,n)}^P$, and the details are presented in Algorithm \ref{algo:gamma_n^P}.
	
	\subsection{User Scheduling Problem}
	\label{subsec:user_scheduling}
	
	In Section \ref{subsec:PowerAllocationProb}, the optimal power allocation ratios are derived for fixed scheduling of user $k$ and user $n$. 
	We then aim at finding better user scheduling,
	and the following lemma gives an insight into user scheduling.
	
	\begin{lemma}
		$\underline{\Delta}_{(k,n)}$ is an increasing function of the received SNR of the SIC user (user $n$) in the high SNR region. 
		Moreover, if $\frac{1-2\gamma_n-2P_t-2\gamma_n P_t}{\gamma_n}>\frac{|h_k|^2}{\sigma^2}$, $\Delta_{(k,n)}$ is an increasing function of the received SNR of the non-SIC user (user $k$); otherwise, $\Delta_{(k,n)}$ is a decreasing function of the received SNR of the non-SIC user.
		\label{lemma:Delta-SNR}
	\end{lemma}
	\begin{proof}
		Since $\gamma_n$, $\gamma_k$, and $\sigma^2$ are constants, the received SNRs of both users depend only on the channel gains.
		Differentiating (\ref{eq:Delta}) with respect to $|h_n|^2$,
		\begin{align}
		\frac{\partial \underline{\Delta}_{(k,n)}}{\partial |h_n|^2} &= \frac{2 \gamma_n(|h_n|^2+\sigma^2)(1-P_t) - \gamma_n|h_n|^2 - \sigma^2}{2 \ln 2 (|h_n|^2+\sigma^2)(\gamma_n|h_n|^2+\sigma^2)} \\
		&= \begin{cases}
		>0 & \textrm{if } \frac{1-2\gamma_n(1-P_t)}{\gamma_n} < \frac{|h_n|^2}{\sigma^2}(1-2P_t) \\
		<0 & \textrm{otherwise}
		\end{cases}.
		\end{align}
		Because the nonextreme value of $\gamma_n > 0$, as $|h_n|^2 \rightarrow \infty$, $\frac{1-2\gamma_n(1-P_t)}{\gamma_n} < \frac{|h_n|^2}{\sigma^2}(1-2P_t)$ is always satisfied, and $\frac{\partial \underline{\Delta}_{(k,n)}}{\partial |h_n|^2}$ converges to zero. 
		Therefore, $\underline{\Delta}_{(k,n)}$ is a nondecreasing function of $|h_n|^2$ in the high-SNR region.
		
		Similarly, 
		\begin{align}
		\frac{\partial \underline{\Delta}_{(k,n)}}{\partial |h_k|^2} &= \frac{-\gamma_n |h_k|^2 + (1-2\gamma_n-2P_t + 2\gamma_n P_t ) \sigma^2}{2\ln 2 (|h_k|^2+\sigma^2)(\gamma_n |h_k|^2 + \sigma^2)} \\
		&=
		\begin{cases}
		>0 & \textrm{if } \frac{1-2\gamma_n-2P_t + 2\gamma_n P_t}{\gamma_n} > \frac{|h_k|^2}{\sigma^2} \\
		<0 & \textrm{otherwise}
		\end{cases}.
		\end{align}
		However, the high-SNR approximation cannot be applied to user $k$ because $|h_k|^2$ is quite small for blind signal classification.
		Recall that $P_k \{\hat{\mathcal{H}}_S|\mathcal{H}_N \}$ increases with the received SNR.
	\end{proof}
	
	
	According to Lemma \ref{lemma:Delta-SNR}, the user with the strongest channel conditions is the best choice for the SIC user (user $n$) because $\underline{\Delta}_{(k,n)}$ increases and $P_n \{ \hat{\mathcal{H}}_N|\mathcal{H}_S \}$ decreases with the received SNR of the SIC user.
	Recall that $|h_1|^2 \geq |h_2|^2 \geq \cdots \geq |h_K|^2$; thus, the optimal choice of the SIC user becomes $n^* =1$.
	On the other hand, the scheduling of the non-SIC user (user $k$) is not clear, because $\underline{\Delta}_{(k,n_0)}$ could be increasing or decreasing depending on the received SNR of the non-SIC user.
	Instead, we can see that the global maximum of $\underline{\Delta}^*_{(k,n)}$ is at $\frac{|h_k|^2}{\sigma^2} = \frac{1 - 2\gamma_n - 2P_t + 2\gamma_n P_t}{\gamma_n}$.
	However, the number of users is finite, and it is almost impossible that a user exists whose SNR is exactly $\frac{1 - 2\gamma_n - 2P_t + 2\gamma_n P_t}{\gamma_n}$; therefore, we will gradually find the better user to be the non-SIC user instead of user $k$.
	The better non-SIC user can be chosen in two ways: the scheduling of the user with the larger or smaller received SNR than the previous scheduled user $k$.
	The next step for the scheduling of the non-SIC user follows one of the following four cases:
	
	\begin{figure}[t!]
		\minipage{0.45\textwidth}
		\includegraphics[width=\linewidth]{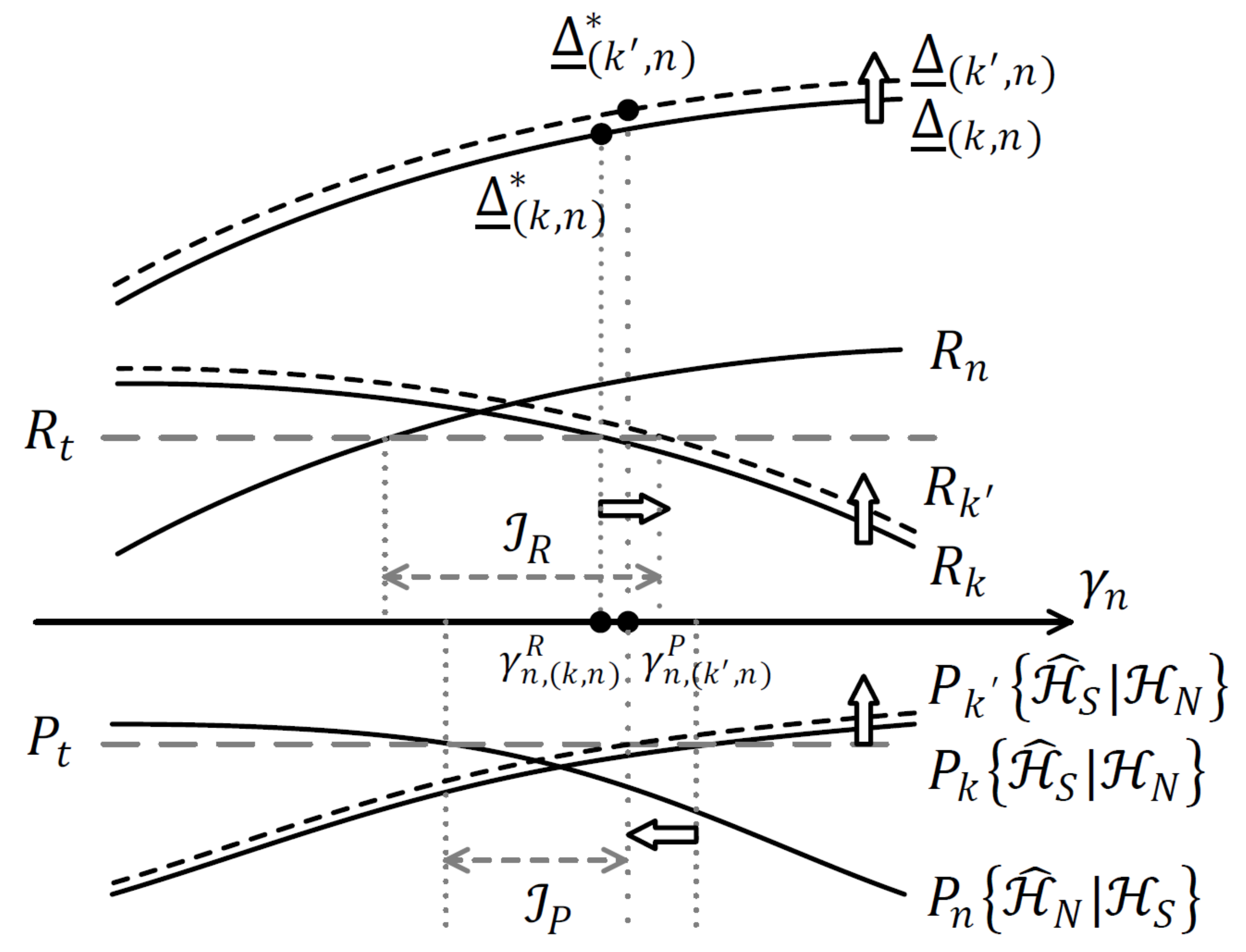}
		\caption{Example of Case 1 of the user scheduling problem} 
		\label{Fig:UserScheduling_case1}
		\endminipage\hfill
		\minipage{0.45\textwidth}
		\includegraphics[width=\linewidth]{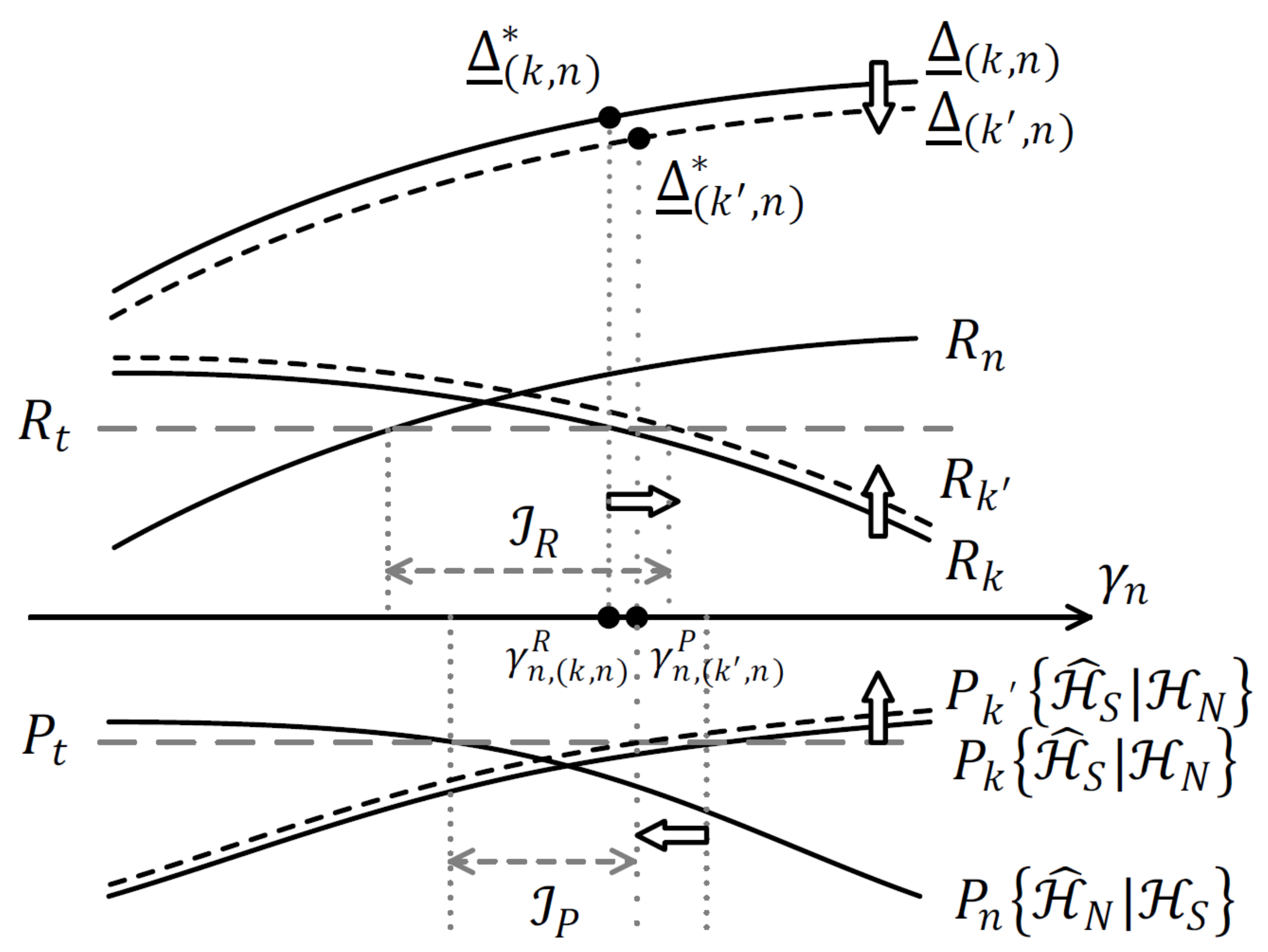}
		\caption{Example of Case 2 of the user scheduling problem}
		\label{Fig:UserScheduling_case2}
		\endminipage
	\end{figure}
	
	\subsubsection{Case 1 ($\gamma_{n,(k,n)}^{*} = \gamma_{n,(k,n)}^R$ and $\frac{1-2\gamma_n-2P_t + 2\gamma_n P_t}{\gamma_n} \geq \frac{|h_k|^2}{\sigma^2}$)}
	
	The minimum data rate constraint is tighter than the classification error constraint, and $\underline{\Delta}_{(k,n)}$ increases with $\frac{|h_k|^2}{\sigma^2}$ by Lemma \ref{lemma:Delta-SNR}. 
	To relieve the constraint of the data rate, the choice of user $k'$ whose SNR is larger than that of user $k$, i.e., $|h_{k'}|^2 > |h_k|^2$, would be beneficial.
	Then, $R_k < R_{k'}$ and $P_k \{ \hat{\mathcal{H}}_S |\mathcal{H}_N \} < P_{k'} \{ \hat{\mathcal{H}}_S|\mathcal{H}_N \}$.
	In Fig. \ref{Fig:UserScheduling_case1}, some arrows and dashed lines indicate which values are changed with the newly chosen non-SIC user, user $k'$.
	$\underline{\Delta}_{(k',n)}$, $R_{k'}$, and $P_{k'} \{\hat{\mathcal{H}}_S|\mathcal{H}_N\}$ increase, and therefore $\gamma_{n,(k'n)}^P < \gamma_{n,(k,n)}^P$ and $\gamma_{n,(k',n)}^R > \gamma_{n,(k,n)}^R$.
	$\mathcal{I_P}$ becomes narrower, but $\mathcal{I_R}$ is wider than before, and this means that the scheduling of user $k'$ rather than user $k$ as the non-SIC user relieves the minimum data rate constraint, but makes the classification error probability constraint tighter.
	In Fig. \ref{Fig:UserScheduling_case1}, when user $k$ is scheduled, $\underline{\Delta}_{(k,n)}^{*}$ is found at $\gamma_{n,(k,n)}^{*} = \gamma_{n,(k,n)}^R$. 
	However, when user $k'$ is scheduled, $\underline{\Delta}_{(k',n)}^{*}$ is obtained when $\gamma_{n,(k',n)}^{*} = \gamma_{n,(k',n)}^P$. 
	Note that $\underline{\Delta}_{(k,n)}$ increases with the received SNR of the non-SIC user and $\gamma_n$.
	Therefore, since $|h_{k'}|^2 > |h_k|^2$ and $\gamma_{n,(k',n)}^{*} > \gamma_{n,(k,n)}^{*}$, we can see that $\underline{\Delta}_{(k,n)} \leq \underline{\Delta}_{(k',n)}$ in Fig. \ref{Fig:UserScheduling_case1}.
	However, it does not guarantee that $\gamma_{n,(k',n)}^{*} > \gamma_{n,(k,n)}^{*}$ always. 
	If the difference in the SNRs of users $k$ and $k'$ is very large, it is possible that $\gamma^*_{n,(k',n)} < \gamma_{n,(k,n)}^*$; then, the maximum classification error constraint would be too tight to guarantee a larger $\underline{\Delta}^*_{(k',n)}$ than $\underline{\Delta}^*_{(k,n)}$. 
	In this case, the comparison step of $\underline{\Delta}^*_{(k',n)}$ and $\underline{\Delta}^*_{(k,n)}$ is necessary.
	If $\underline{\Delta}^*_{(k',n)} > \underline{\Delta}^*_{(k,n)}$, user $k'$ is scheduled as the non-SIC user; otherwise, user $k$ remains as the non-SIC user.
	
	\subsubsection{Case 2 ($\gamma_{n,(k,n)}^{*} = \gamma_{n,(k,n)}^R$ and $\frac{1-2\gamma_n-2P_t - 2\gamma_n P_t}{\gamma_n} < \frac{|h_k|^2}{\sigma^2}$)}
	
	The minimum data rate constraint is tighter than the constraint of classification; therefore, user $k'$ with a better SNR than user $k$, i.e., $|h_{k'}|^2 > |h_k|^2$, is chosen similar to Case 1.
	The changes in $R_{k'}$, $P_{k'} \{ \hat{\mathcal{H}}_S|\mathcal{H}_N \}$, $\gamma_{n,(k',n)}^R$, and $\gamma_{n,(k',n)}^P$ are all the same as those in Case 1.
	However, $\underline{\Delta}_{(k,n)}$ decreases with $|h_k|^2$ in this case; therefore $\underline{\Delta}_{(k,n)}^{*} < \underline{\Delta}_{(k',n)}^{*}$ cannot be guaranteed.
	As shown in Fig. \ref{Fig:UserScheduling_case2}, even though $\gamma_{n,(k',n)}^{*}=\gamma_{n,(k',n)}^P$ becomes larger than $\gamma_{n,(k,n)}^{*} = \gamma_{n,(k,n)}^R$, it is possible that $\underline{\Delta}_{(k,n)}^{*} > \underline{\Delta}_{(k',n)}^{*}$. 
	However, if the slope of the $\underline{\Delta}_{(k,n)}$ graph is very steep, $\underline{\Delta}_{(k,n)}^{*} < \underline{\Delta}_{(k',n)}^{*}$ is also possible. 
	Therefore, a comparison of the newly updated $\underline{\Delta}_{(k',n)}^{*}$ with the previously obtained $\underline{\Delta}_{(k,n)}^{*}$ is necessary.
	If $\underline{\Delta}^{*}_{(n,k)} < \underline{\Delta}^{*}_{(n,k')}$, user $k'$ is preferred as the non-SIC user rather than user $k$.
	Otherwise, user $k$ is determined as the non-SIC user.
	
	\begin{figure}[t!]
		\minipage{0.45\textwidth}
		\includegraphics[width=\linewidth]{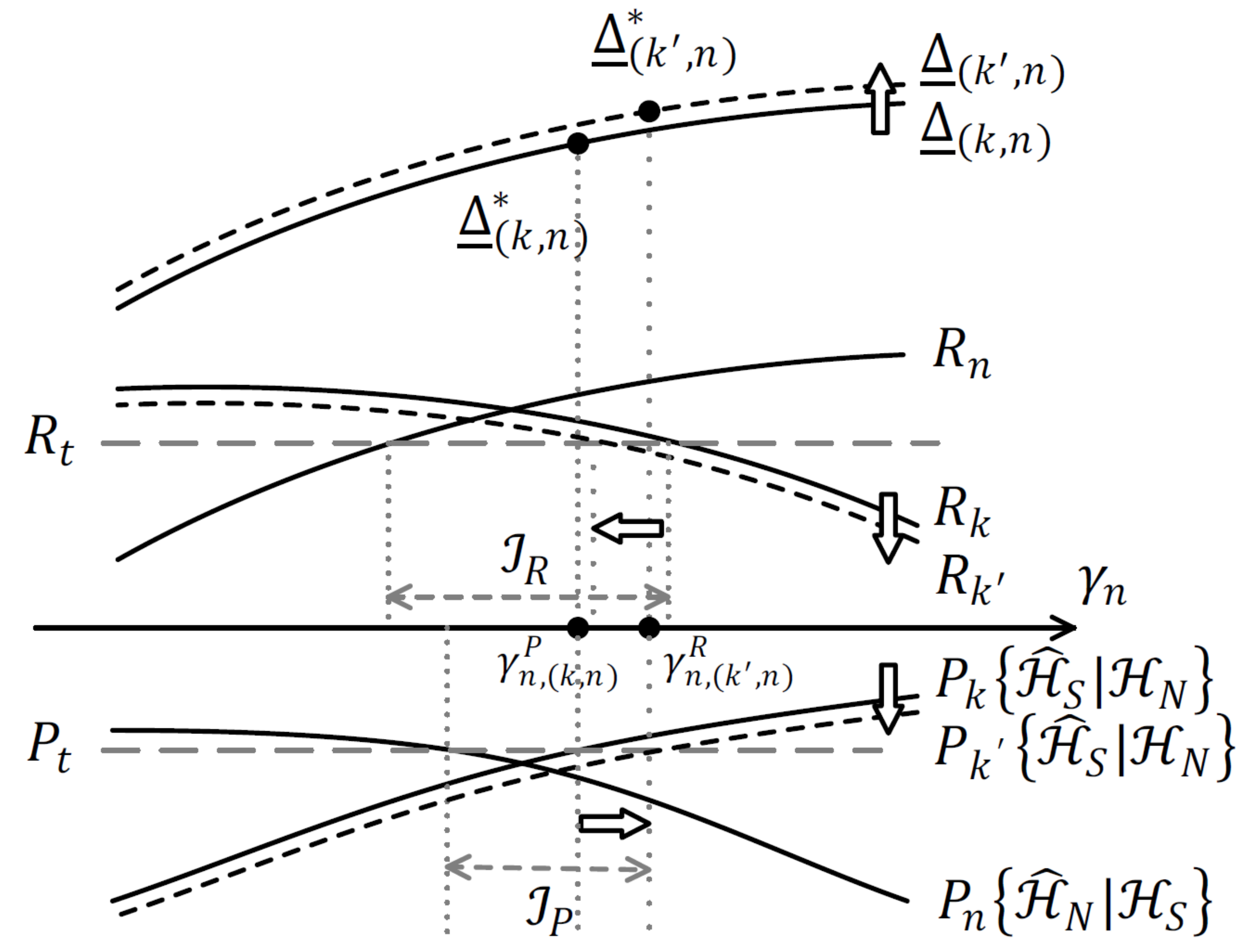}
		\caption{Example of Case 3 of the user scheduling problem} 
		\label{Fig:UserScheduling_case3}
		\endminipage\hfill
		\minipage{0.45\textwidth}
		\includegraphics[width=\linewidth]{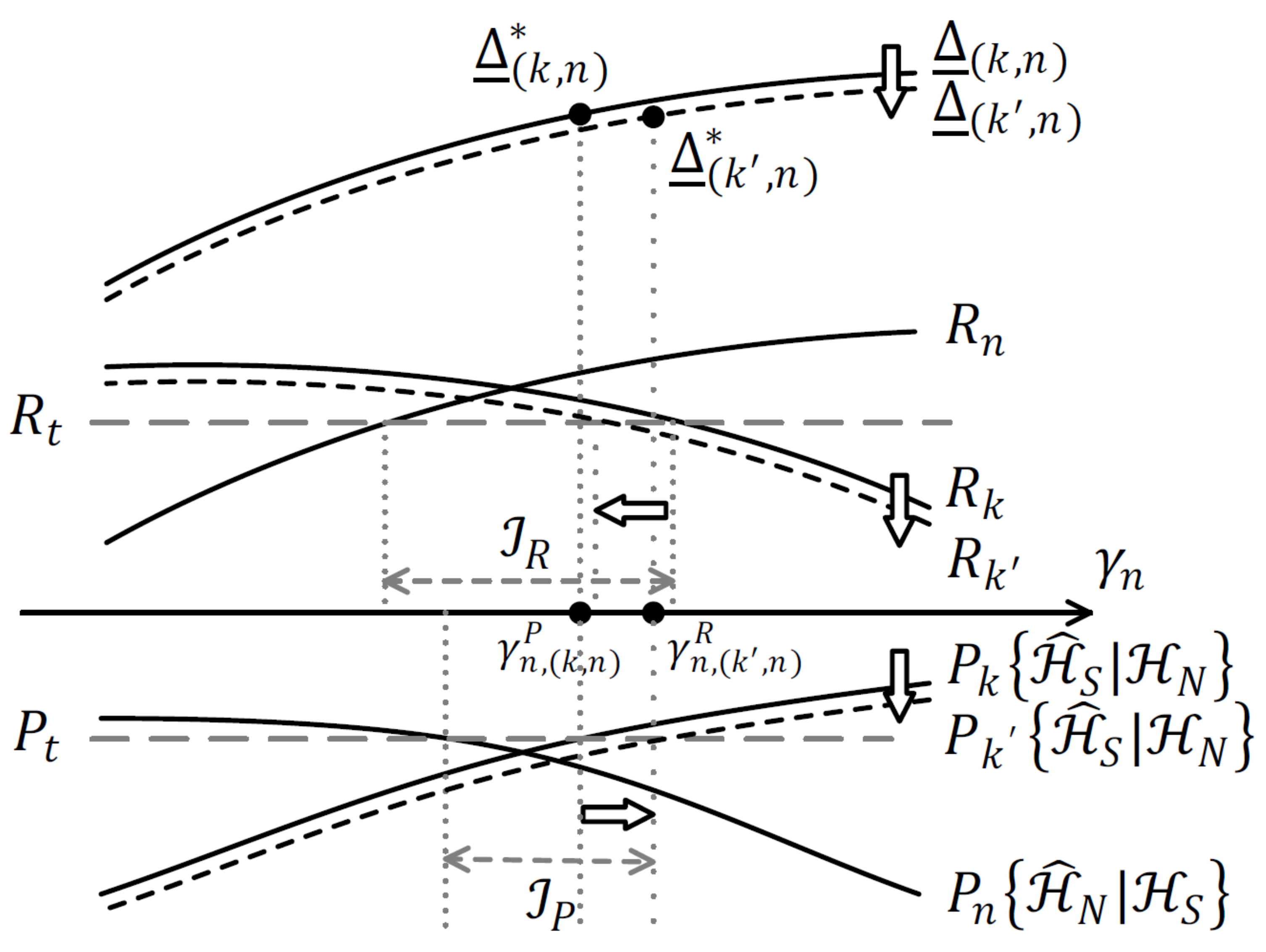}
		\caption{Example of Case 4 of the user scheduling problem}
		\label{Fig:UserScheduling_case4}
		\endminipage
	\end{figure}
	
	\subsubsection{Case 3 ($\gamma_{n,(k,n)}^{*} = \gamma_{n,(k,n)}^P$ and $\frac{1-2\gamma_n-2P_t -2\gamma_n P_t}{\gamma_n} < \frac{|h_k|^2}{\sigma^2}$)}
	
	In this case, the classification error constraint is tighter than the minimum data rate constraint.
	Therefore, the BS takes user $k'$ with a smaller SNR rather than user $k$ as the non-SIC user; i.e., $|h_{k'}|^2 < |h_k|^2$.
	$P_k \{ \hat{\mathcal{H}}_S |\mathcal{H}_N \}$ and $R_k$ then decrease, and therefore $\mathcal{I}_P$ is enlarged but $\mathcal{I}_R$ becomes narrower.
	Fig. \ref{Fig:UserScheduling_case3} shows that $\gamma_{n,(k,n)}^{*} = \gamma_{n,(k,n)}^P$ but $\gamma_{n,(k',n)}^{*} = \gamma_{n,(k',n)}^R$.
	Note that $\underline{\Delta}_{(k,n)}$ decreases with the received SNR of the non-SIC user but increases with $\gamma_n$.
	Since $|h_{k'}|^2 < |h_k|^2$ and $\gamma_{n,(k,n)}^{*} < \gamma_{n,(k',n)}^{*}$ in Fig. \ref{Fig:UserScheduling_case3}, $\underline{\Delta}_{(k,n)}^{*} < \underline{\Delta}_{(k',n)}^{*}$ obviously, and the scheduling of the non-SIC user with the smaller SNR is beneficial in terms of the sum-rate gain of NOMA over OMA.
	However, if $\gamma_{n,(k,n)}^* > \gamma_{n,(k,n)}^*$, the scheduling of user $k'$ is only conducted when $\underline{\Delta}_{(k,n)}^{*} < \underline{\Delta}_{(k',n)}^{*}$, similar to Case 1.
	
	\subsubsection{Case 4 ($\gamma_{n,(k,n)}^{*} = \gamma_{n,(k,n)}^P$ and $\frac{1-2\gamma_n-2P_t - 2\gamma_n P_t}{\gamma_n} \geq \frac{|h_k|^2}{\sigma^2}$)}
	
	The maximum classification error constraint is tighter than the minimum data rate constraint; therefore, user $k'$ with the weaker SNR is chosen rather than user $k$. 
	The changes in the parameters are the same as those in Case 3, except that $\underline{\Delta}_{(k,n)}$ decreases with the received SNR of the non-SIC user.
	Therefore, it cannot be guaranteed that user $k'$ is better than user $k$ as the non-SIC user, and it is necessary to compare $\underline{\Delta}^{*}_{(k,n)}$ with $\underline{\Delta}^{*}_{(k',n)}$ after replacing user $k$ by user $k'$.
	If  $\underline{\Delta}^{*}_{(k,n)} < \underline{\Delta}^{*}_{(k',n)}$, user $k'$ is scheduled;
	otherwise, user $k$ is more appropriate as the non-SIC user than user $k'$.
	
	By iteratively examining users in increasing or decreasing order of the SNR, depending on each case, the BS can find the most appropriate one as the non-SIC user giving the maximum value of $\underline{\Delta}_{(k_0,n_0)}^{*}$.
	In addition, the optimal power allocation rule is obtained simultaneously.
	The details are presented in Algorithm \ref{algo}.
	
	\begin{algorithm}
		\caption{Iterative algorithm for jointly optimizing user scheduling and power allocation ratios
			\label{algo}}
		\begin{algorithmic}[1]
			\Require{\\
				\begin{itemize}
					\item $|h_1|^2 \geq \cdots \geq |h_K|^2$: Channel gains of $K$ users
					\item $R_t$: minimum threshold of the data rate 
					\item $P_t$: maximum threshold of the classification error probability
					\item $\sigma^2$: normalized noise variance
				\end{itemize}
			}
			\State{Choose $n^* = 1$}
			\State{Arbitrarily choose $k^* = k$, for $k \in \{2,\cdots,K\}$}
			
			\While{True}{
				
				\State{Compute $\gamma_{n^*,(k^*,n^*)}^R$, $\gamma_{n^*,(k^*,n^*)}^P$, and $\underline{\Delta}_{(k^*,n^*)}$}
				\If{$\gamma_{n,(k^*,n^*)}^R < \gamma_{n,(k^*,n^*)}^P$}
				\State{$k' = k^* + 1$}
				\If{$k' == K+1$} 
				\State{$\underline{\Delta}_{(k^*,n^*)}^{*} = \underline{\Delta}_{(k^*,n^*)}$, $\gamma_{(k^*,n^*)}^{*} = \gamma_{(k^*,n^*)}^R$} 
				\State{break;}
				\EndIf
				\State{Compute $\gamma_{n,(k^*,n)}^*$ and $\gamma_{n,(k',n)}^*$ by Theorem \ref{thm:optPower}}
				\If{$\frac{1-2\gamma_{n,(k,n^*)}^R-2P_t-2\gamma_{n,(k,n^*)}^R P_t}{\gamma_{n,(k,n^*)}^R} \geq \frac{|h_k|^2}{\sigma^2}$  \textbf{and} $\gamma_{n,(k^*,n)}^* < \gamma_{n,(k',n)}^*$}
				\State{flag = False;}
				\Else \State{flag = True;}
				\EndIf
				
				\Else
				\State{$k' = k^* - 1$}
				\If{$k' == 1$} 
				\State{$\underline{\Delta}_{(k,n)}^{*} = \underline{\Delta}_{(k,n)}$, $\gamma_{(k,n)}^{*} = \gamma_{(k,n)}^P$} 
				\State{break;}
				\EndIf
				\If{$\frac{1-2\gamma_{n,(k,n^*)}^R-2P_t-2\gamma_{n,(k,n^*)}^R P_t}{\gamma_{n,(k,n^*)}^R} < \frac{|h_k|^2}{\sigma^2}$  \textbf{and} $\gamma_{n,(k^*,n)}^* < \gamma_{n,(k',n)}^*$}
				\State{flag = False;}
				\Else
				\State{flag = True;}
				\EndIf
				\EndIf
				
				\If{flag  \textbf{and}  $\underline{\Delta}_{(k^*,n^*)} > \underline{\Delta}_{(k',n^*)}$}
				\State{$\underline{\Delta}_{(k^*,n^*)}^{*} = \underline{\Delta}_{(k_0,n^*)}$}
				\State{break;}
				\Else
				\State{$k^* = k'$}
				\EndIf
				
			}
			\EndWhile
		\end{algorithmic}
	\end{algorithm}

	\section{Numerical Results}
	\label{sec:simulation}
	
	\subsection{Blind Classification Error Probability}
	
	In Section \ref{sec:ergodic_capacity_analysis}, the blind classification error probabilities were analytically determined, when the BS knows the instantaneous channel gains.
	In Fig. \ref{Fig:ClassificationErrorProb}, we compare the classification error probabilities obtained by the analytical method and Monte Carlo simulation.
	Assume that the non-SIC user and SIC user are modulated by QPSK and 16-QAM, respectively, and that $\gamma_n = 0.24$.
	Further, the $L=1$ data sample is used for blind signal classification.
	Fig. \ref{Fig:ClassificationErrorProb} shows the plots of two classification error types, $P \{ \hat{\mathcal{H}}_N|\mathcal{H}_S \}$ and $P \{\hat{\mathcal{H}}_S|\mathcal{H}_N \}$, versus the received SNR.
	As we noted before, the probability of the SIC user being incorrectly classified as the non-SIC user, $P \{\hat{\mathcal{H}}_N|\mathcal{H}_S \}$, decreases as the SNR increases, but the probability of the non-SIC user classifying itself as the SIC user, $P \{ \hat{\mathcal{H}}_S|\mathcal{H}_N \}$, increases with the received SNR.
	If $P_t = 0.1$, the SNR of the non-SIC user would be smaller than 5 dB, but that of the SIC user should be larger than approximately 23 dB, as shown in Fig. \ref{Fig:ClassificationErrorProb}.
	In addition, the analytical results of $P \{\hat{\mathcal{H}}_N | \mathcal{H}_S \}$ are very similar to the simulation results.
	The analytical results are based on (\ref{eq:P_1-final}); thus, we can conclude that the approximation used in (\ref{eq:P_1-final}) is somewhat reliable.
	The graphs of $P \{\hat{\mathcal{H}}_S | \mathcal{H}_N \}$ show a small difference, but they are still similar. 
	As mentioned in Section \ref{sec:ergodic_capacity_analysis}, the high-SNR approximation cannot be applied to compute $P \{\hat{\mathcal{H}}_S | \mathcal{H}_N \}$, but we further approximate it by (\ref{eq:P_2-approx_bf}). This approximation explains the differences between the analytical and simulation results.
	
	\begin{figure} [t!]
		\centering
		\includegraphics[width=0.42\textwidth]{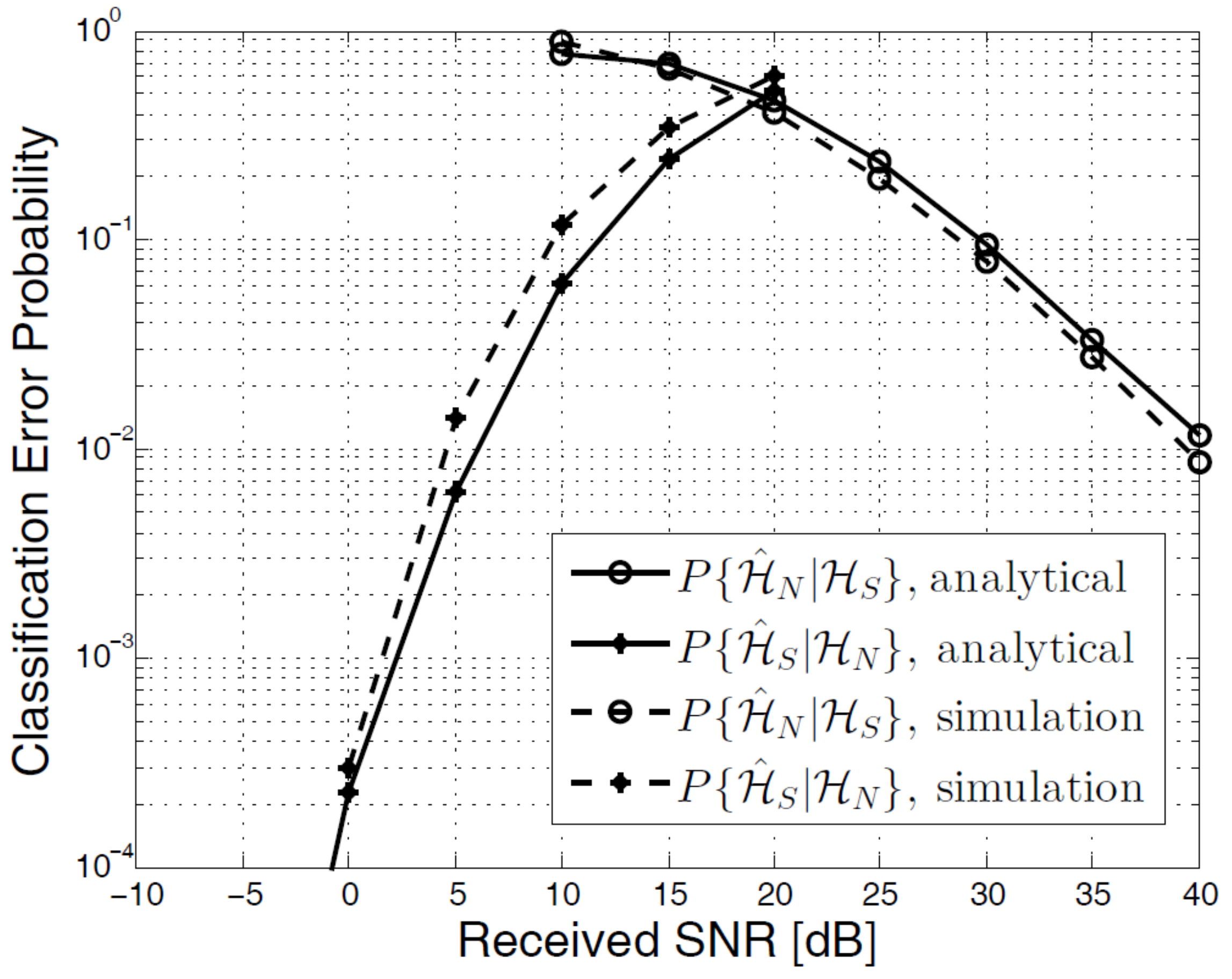}
		\caption{Classification error probabilities obtained by the analytical method and simulation}
		\label{Fig:ClassificationErrorProb}
	\end{figure}
	
	\subsection{Sum-Rate Gain of NOMA over OMA}
	
	To verify that Algorithm \ref{algo} works well for the joint optimization problem of user scheduling and power allocation for NOMA transmissions, the sum-rate gain of NOMA over OMA is compared with those of other scheduling methods.
	Consider the cellular downlink model in Fig. \ref{Fig:SystemModel}.
	$K=40$ users are randomly located in the circular region of radius $r=50$ centered at the BS.
	The BS chooses two users among $K$ users for NOMA transmissions to provide an increased data rate compared to OMA while guaranteeing the minimum data rate constraint and reliable blind signal classification.
	For the simulation, $R_t = 0.8$, $P_t = 10^{-2}$, a transmit SNR of 10 dB, and $L=5$ are used unless otherwise noted.
	
	According to \cite{NOMA-UserSchedule:Ding:TVT2016}, a pair of users having the strongest and weakest channel gains is the best in an F-NOMA system. 
	On the other hand, the performance gain of NOMA over OMA could be maximized when the users with the strongest and second-strongest channel gains are paired for a CR-NOMA system.
	In CR-NOMA, it is important to guarantee the weak user's QoS, and thus it is consistent with the minimum data rate constraint in (\ref{eq:const-minimum-rate}). 
	We compare those two methods with the proposed algorithm.
	Along with the proposed algorithm, the user with the strongest channel condition, i.e., user $n^*=1$, is always chosen as the SIC user for those comparison schemes; hence, what remains is to find the appropriate non-SIC user.
	Since \cite{NOMA-UserSchedule:Ding:TVT2016} does not consider blind signal classification, we assume that the comparison schemes find the non-SIC user in the user set $\mathcal{U}_P$, whose members can guarantee both the constraints of the minimum data rate and maximum classification error probability.
	In other words, when user $n^*$ and any user $k \in \mathcal{U}_P$ are scheduled as the SIC user and non-SIC user, respectively, $\mathcal{I}_R \cap \mathcal{I}_P \neq \phi$ is satisfied, and there exists a $\gamma_n \in [0,1]$ satisfying both the constraints of the data rate and classification performance.
	In summary, the comparison methods of user scheduling are explained as follows:
	\subsubsection{Strongest--Strongest}
	The BS schedules the user having the second-strongest channel as the non-SIC user among the user set $\mathcal{U}_P$.
	If $\mathcal{U}_P$ includes all $K$ users, $n^* = 1$, and $k^* = 2$.
	
	\subsubsection{Strongest--Weakest}
	The BS schedules the user having the weakest channel gains as the non-SIC user among the user set $\mathcal{U}_P$.
	If $\mathcal{U}_P$ includes all $K$ users, $n^* = 1$, and $k^* = K$.
	
	The above comparison methods also find the optimal power allocation ratio by Theorem \ref{thm:optPower} after user scheduling.
	However, the main difference between the proposed algorithm and the above two methods is the dependency between the problems of user scheduling and power allocation.
	The steps of user scheduling and finding the power allocation for NOMA are independent in the above comparison methods, but the proposed algorithm iteratively finds the appropriate non-SIC user with consideration of the optimal power allocation corresponding to the iterative decision of user scheduling.
	
	Fig. \ref{Fig:CapGain_SNRdBList} shows plots of the sum-rate gain of NOMA over OMA versus the transmit SNR.
	We can easily see that the proposed algorithm gives the best performance compared to ``Strongest--Strongest" and ``Strongest--Weakest."
	The sum-rate gains of the proposed algorithm and ``Strongest--Weakest" increase as the transmit SNR increases, but that of ``Strongest-Strongest" maintains almost the same value for all transmit SNRs. 
	Since ``Strongest--Strongest" chooses the user with the second-strongest channel gain as the non-SIC user, 
	the non-SIC user is highly likely to not satisfy the classification error constraint.
	Thus, the sum-rate of ``Strongest--Strongest" does not increase with the transmit SNR.
	
	\begin{figure} [t!]
		\centering
		\includegraphics[width=0.42\textwidth]{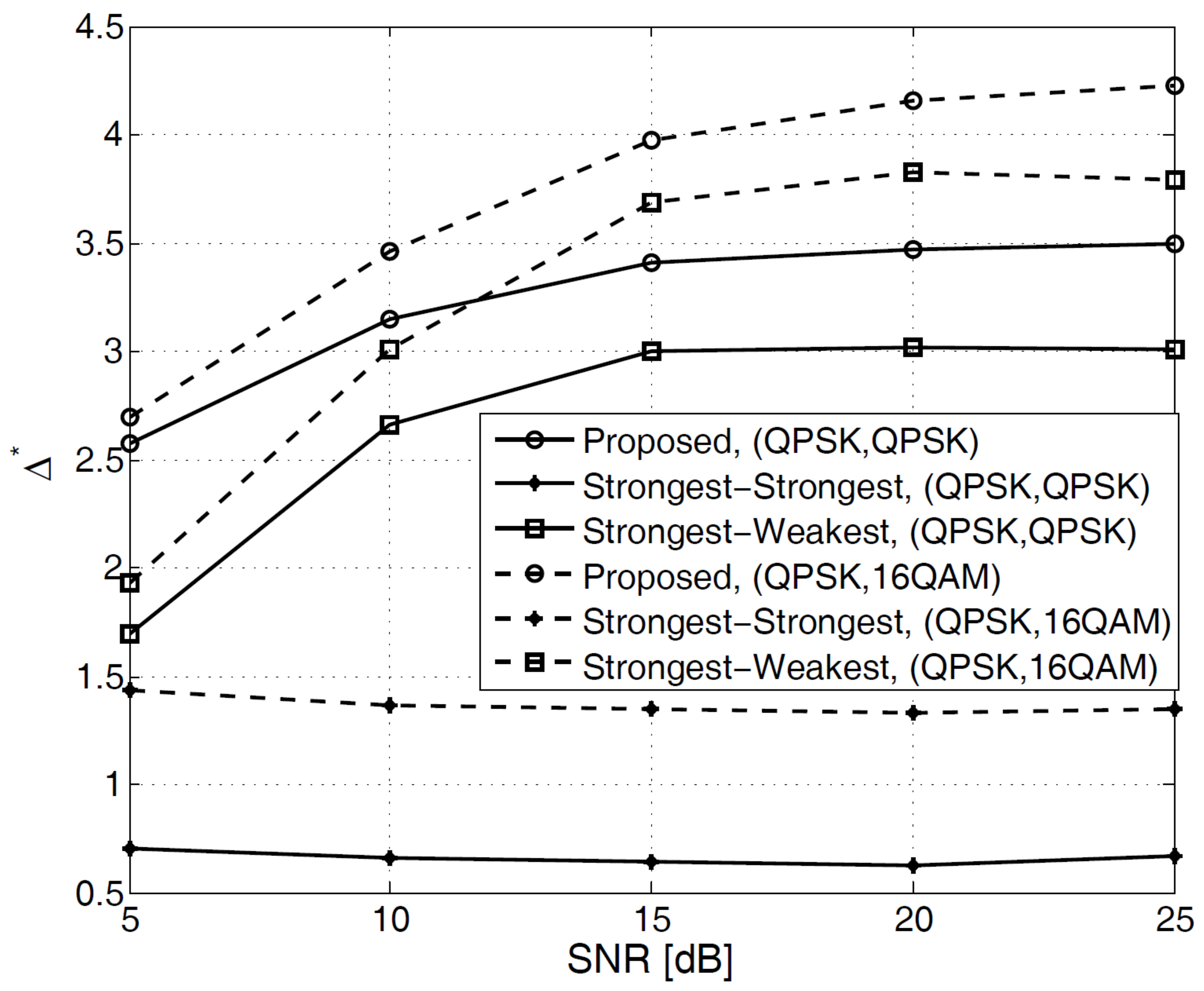}
		\caption{Sum-rate gain of NOMA over OMA versus the transmit SNR}
		\label{Fig:CapGain_SNRdBList}
	\end{figure}
	
	Meanwhile, the performances of the proposed algorithm and that of ``Strongest--Weakest" tend to converge as the SNR increases.
	The high transmit SNR is beneficial for the data rate but not for classification of the non-SIC user, $P_k\{\hat{\mathcal{H}}_N | \mathcal{H}_S\}$.
	This means that the classification error constraint is tighter than the minimum data rate constraint in the high-SNR region, and thus the optimal solution is heavily dominated by the classification performance.
	This situation corresponds to Case 3 in Section \ref{subsec:user_scheduling} because the tight constraint of the classification error makes $\gamma_{n,(k,n)}^* = \gamma_{n,(k,n)}^P$, and the high SNR usually satisfies $\frac{1-2\gamma_n-2P_t-2\gamma_n P_t}{\gamma_n} < \frac{|h_k|^2}{\sigma^2}$. 
	In Case 3, the scheduling of the user with a lower SNR as the non-SIC user is preferred, and therefore the performance gap between the proposed method one and ``Strongest--Weakest" decreases as the SNR increases.
	If the transmit SNR is larger than a certain threshold, the number of users satisfying the classification error constraint, i.e., $|\mathcal{U}_P|$, is reduced, and $\underline{\Delta}^*$ of ``Strongest--Weakest" decreases with SNR$>$20 dB in Fig. \ref{Fig:CapGain_SNRdBList}.
	On the other hand, the proposed scheme shows a tendency to increase more, even with a high SNR, compared to ``Strongest--Weakest," because it is still possible for Case 2 to depend on the random fading channel, although most situations with a high SNR correspond to Case 3.
	In other words, the high SNR still satisfies $\frac{1-2\gamma_n-2P_t-2\gamma_n P_t}{\gamma_n} < \frac{|h_k|^2}{\sigma^2}$, but a deep fade makes the optimal power allocation dominated by the minimum data constraint, $\gamma_{n,(k,n)}^* = \gamma_{n,(k,n)}^R$.
	In Case 2, the user with a higher SNR is preferred for the non-SIC user, and therefore the proposed algorithm can find better user scheduling than ``Strongest--Weakest."
	
	On the other hand, the performance gap between the proposed scheme and ``Strongest--Weakest" increases in the low-SNR region. 
	The low SNR is likely to satisfy $\frac{1-2\gamma_n-2P_t-2\gamma_n P_t}{\gamma_n} > \frac{|h_k|^2}{\sigma^2}$. 
	Moreover, since both $R_k$ and $P_k\{\hat{\mathcal{H}}_N|\mathcal{H}_S\}$ decrease as the SNR decreases, the optimal solution is usually dominated by the minimum data rate constraint, unlike the situation in the high-SNR region.
	Accordingly, most situations in the low SNR region correspond to Case 1. 
	In this case, the scheduling of the user with stronger channel conditions as the non-SIC user makes $\underline{\Delta}^*$ larger.
	Thus, a larger performance gain between the proposed scheme and ``Strongest--Weakest" is obtained with lower transmit SNRs.
	
	We also observe the effect of modulation in Fig. \ref{Fig:CapGain_SNRdBList}, where the solid and dashed lines indicate the results obtained when the SIC user is modulated by 16-QAM and QPSK, respectively, and QPSK is used for the non-SIC user.
	The performance trends of the comparison technologies are generally similar, but it can be seen that the improvement in the data rate with increasing SNR is less when the SIC user is modulated by QPSK compared to the case of 16-QAM.
	According to (\ref{eq:ML_SICuser}) and (\ref{eq:ML_nonSICuser}), the classification of the SIC user becomes easier, and the classification of the non-SIC user becomes more difficult when the composite constellation $\chi$ and non-SIC user's constellation $\chi_k$ are more clearly distinguishable from one another.
	A large minimum distance between $\chi$ and $\chi_k$ can be obtained with the small modulation order of the SIC user; thus, $P_n \{ \hat{\mathcal{H}}_N | \mathcal{H}_S \}$ is lower, but $P_k \{ \hat{\mathcal{H}}_S | \mathcal{H}_N \}$ is higher with QPSK for the SIC user.
	The higher $P_k \{ \hat{\mathcal{H}}_S | \mathcal{H}_N \}$ makes the classification error constraint tighter, and hence the advantage of increasing the SNR is not shown clearly when QPSK is used for the SIC user compared to the case of 16-QAM.
	
	\subsection{System Parameters $L$ and $P_t$}
	
	Figs. \ref{Fig:CapGain_sampleList} and \ref{Fig:CapGain_PtList} show plots of $\underline{\Delta}^*$ versus the number of data samples of $L$ and $P_t$, respectively.
	We can easily see that $\underline{\Delta}^*$ of ``Strongest--Strongest" decreases with $L$, which means that blind signal classification does not work well in ``Strongest--Strongest" owing to the high SNR of the non-SIC user; i.e., $P_k \{ \hat{\mathcal{H}}_S | \mathcal{H}_N \} > 0.5$.
	On the other hand, $\Delta^*$ of the proposed scheme and ``Strongest--Weakest" increases as $L$ grows, and therefore those user scheduling techniques can conduct blind signal classification well.
	The sum-rate gains of the proposed method and ``Strongest--Weakest" converge, and so Fig. \ref{Fig:CapGain_sampleList} can be used to determine how many data samples are enough for blind signal classification. 
	When QPSK and 16-QAM are used for the non-SIC user and SIC user, respectively, $L=3$ is sufficient to boost the sum-rate gain. 
	Thus, excessive complexity is not necessary for blind signal classification.
	
	\begin{figure}[t!]
		\minipage{0.45\textwidth}
		\includegraphics[width=\linewidth]{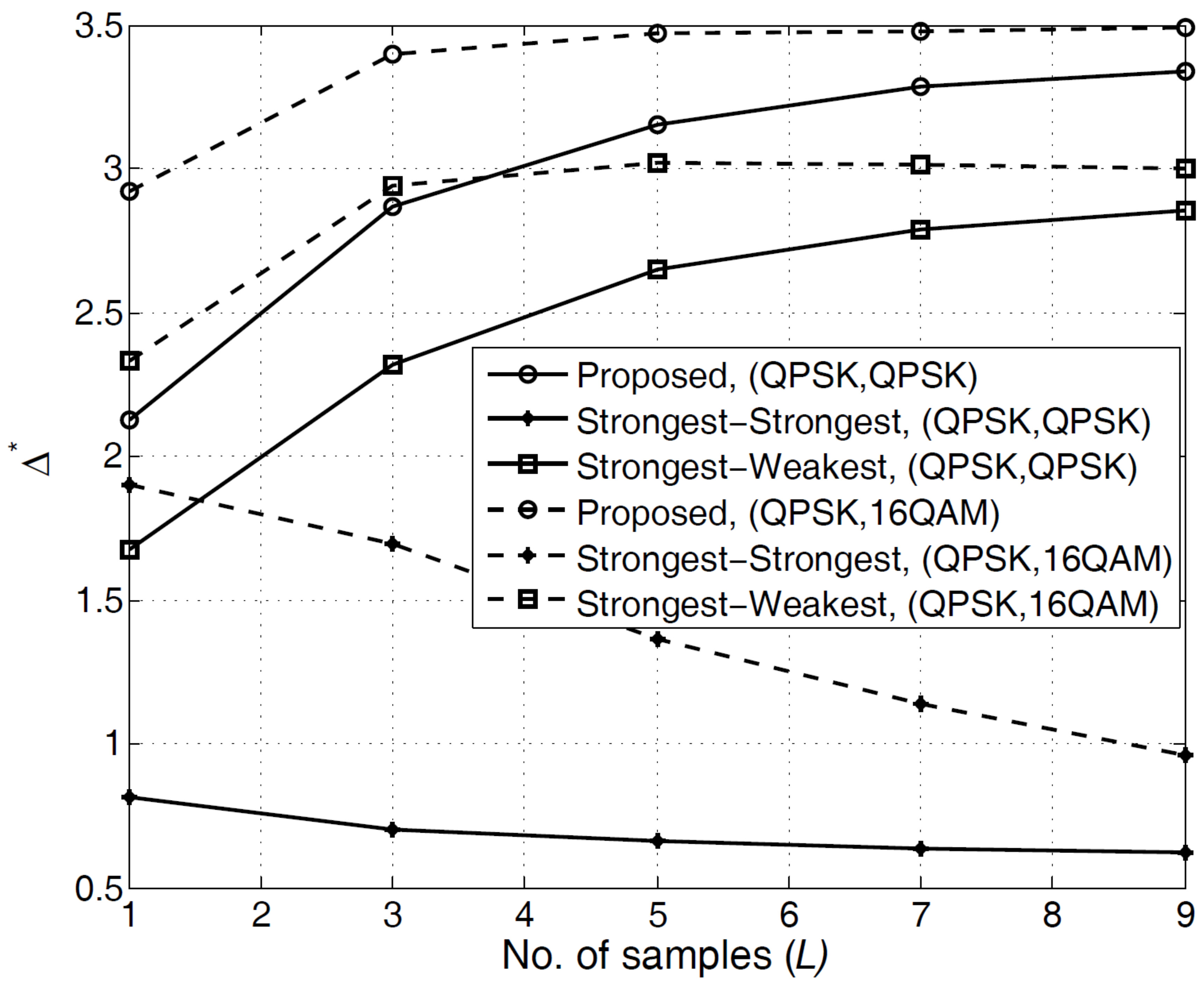}
		\caption{Sum-rate gain of NOMA over OMA versus the number of data samples} 
		\label{Fig:CapGain_sampleList}
		\endminipage\hfill
		\minipage{0.45\textwidth}
		\includegraphics[width=\linewidth]{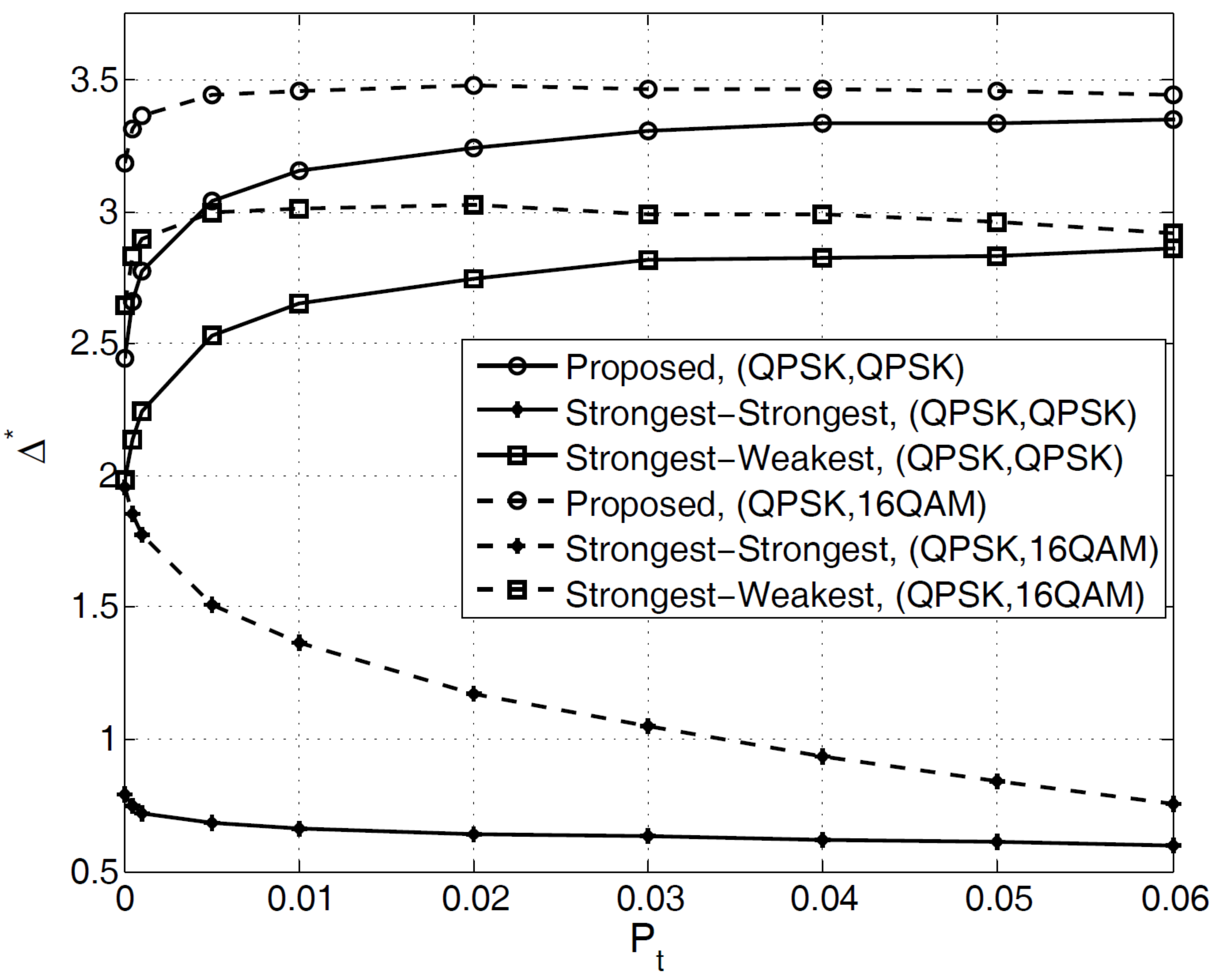}
		\caption{Sum-rate gain of NOMA over OMA versus $P_t$}
		\label{Fig:CapGain_PtList}
		\endminipage
	\end{figure}
	
	In addition, the system designer could choose the appropriate value of $P_t$ by observing the trade-off between the data rate and the performance of blind signal classification.
	As mentioned before, if $P_t$ is too small, the non-SIC user with a very low SNR is scheduled to satisfy $P_t$; therefore, the data rate decreases. 
	On the other hand, if the classification constraint becomes loose, i.e., $P_t$ is large, incorrect signal classification occurs more frequently, and the data rate slowly decreases.
	From Fig. \ref{Fig:CapGain_PtList}, the appropriate value of $P_t$ can be found; e.g., $P_t=0.01$ when QPSK and 16-QAM are used for the non-SIC user and SIC user, respectively.
	Note that even though the sum-rate gain of NOMA over OMA remains constant for a certain interval of $P_t$, a small $P_t$ is preferred to minimize the occurrence of classification failure.
	
	\section{Conclusion}
	\label{sec:conclusion}
	
	This paper mainly focuses on the blind classification of NOMA signals for the presence of interference to determine whether or not a user should perform SIC, without any useful high-layer signaling.
	The error probabilities of blind signal classification are mathematically derived for the cases when the SIC user is classified as the non-SIC user and when the non-SIC user classifies itself as the SIC user.
	On the basis of analytical results, we formulate the joint optimization problem of user scheduling and power allocation to maximize the sum-rate gain of NOMA over OMA subject to the constraints of the maximum classification error probability and minimum data rate.
	To solve this problem, we investigate the effects of blind signal classification on user scheduling and power allocation for NOMA users.
	An iterative algorithm is then proposed for user scheduling and for finding the power allocation ratios of the scheduled users.
	Finally, numerical results confirm that the proposed scheme gives better sum-rate gains over OMA compared to conventional user scheduling methods.

	\begin{IEEEbiography}[{\includegraphics[width=1in,height=1.25in,clip,keepaspectratio]{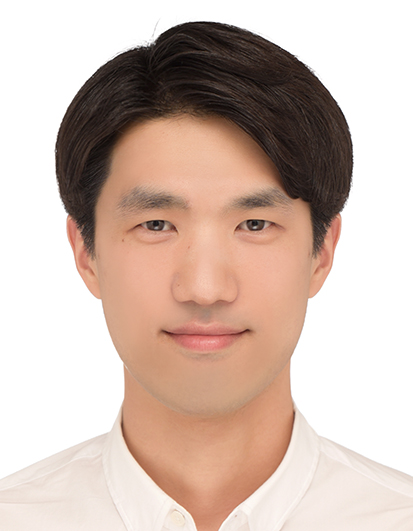}}]{Minseok Choi}
		received the BS and MS degrees in electrical engineering from the Korea Advanced Institute of Science and Technology (KAIST), Daejeon, Republic of Korea, in 2011 and 2013, respectively. He is currently pursuing the PhD degree in KAIST. His research interests include wireless caching network, NOMA, 5G communications, and stochastic network optimization.\end{IEEEbiography}
	\begin{IEEEbiography}[{\includegraphics[width=1in,height=1.25in,clip,keepaspectratio]{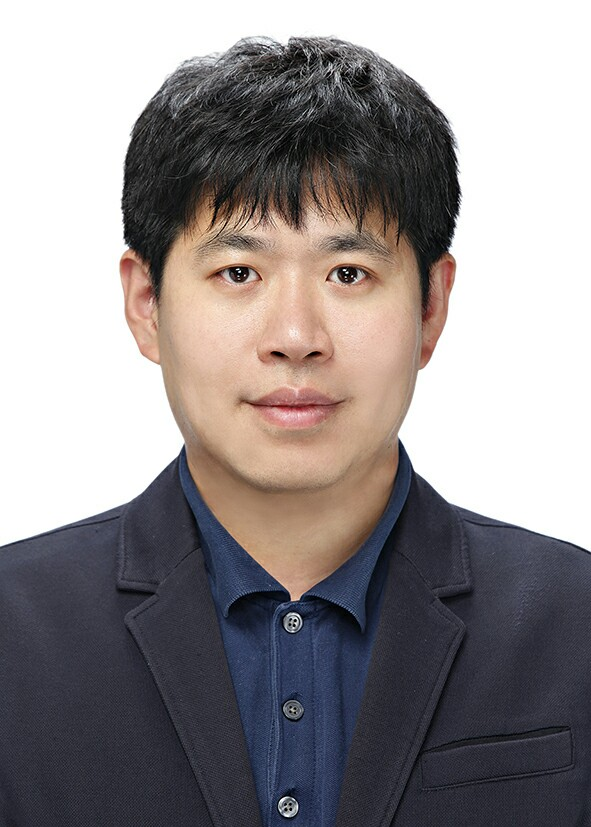}}]{Joongheon Kim}
		(M'06--SM'18) has been an assistant professor with Chung-Ang University, Seoul, Korea, since 2016. He received his BS (2004) and MS (2006) degrees from Korea University, Seoul, Korea; and his PhD (2014) degree from the University of Southern California (USC), Los Angeles, CA, USA. In industry, he was with the LG Electronics Seocho R\&D Campus (Seoul, Korea, 2006--2009), InterDigital (San Diego, CA, USA, 2012), and Intel Corporation (Santa Clara, CA, USA, 2013--2016). 
		
		He is a senior member of the IEEE; and a member of IEEE Communications Society. He was awarded the Annenberg Graduate Fellowship with his PhD admission from USC (2009).
	\end{IEEEbiography}\vfill
\end{document}